\documentclass{article}

\usepackage{geometry}
\usepackage{amsthm}
\usepackage{amssymb}
\usepackage{amsfonts}
\usepackage{tikz}
\usepackage{epstopdf}
\usepackage{authblk}
\usepackage{amsmath,amscd}
\usepackage{float}
\usepackage[caption = false]{subfig}
\usepackage{pgfplots}
\usepackage{tikz}
\usepackage{amscd}
\usepackage{graphicx}
\usepackage{color}
\usepackage{hyperref}
\usepackage{bold-extra}
\usepackage{algorithm2e}
\usepackage{dcolumn}
\usepackage[shortlabels]{enumitem}

\newtheorem{lemma}{Lemma}[section]

\newtheorem{prop}{Proposition}[section]

\newcolumntype{d}[1]{D{.}{.}{#1}}
\newcolumntype{m}[1]{D{\pm}{\pm}{#1}}

\newcommand{\R}{\mathbb{R}}
\newcommand{\G}{\mathcal{G}}
\newcommand{\T}{\mathcal{T}}
\newcommand{\C}{\mathcal{C}}
\newcommand{\Po}{\mathcal{P}}
\newcommand{\bfx}{\boldsymbol{x}}
\newcommand{\bfy}{\boldsymbol{y}}
\newcommand{\bfz}{\boldsymbol{z}}
\newcommand{\bfu}{\boldsymbol{u}}
\newcommand{\bfw}{\boldsymbol{w}}
\newcommand{\bfv}{\boldsymbol{v}}
\newcommand{\bfr}{\boldsymbol{r}}
\newcommand{\e}[1]{^{(#1)}} 

\definecolor{todoGreen}{rgb}{0.0, 0.5, 0.0}
\begin{document}

\title{Calculating the symmetry number of flexible sphere clusters}

\author[1]{Emilio Zappa}
\author[2]{Miranda Holmes-Cerfon}

\affil[1]{Mathematics Department, Fordham University, NY}
\affil[2]{Courant Institute of Mathematical Sciences, New York University, NY.}

\date{}
\maketitle

\begin{abstract}
We present a theoretical and computational framework to compute the symmetry number of a flexible sphere cluster in $\mathbb{R}^3$, using a definition of symmetry that arises naturally when calculating the equilibrium probability of a cluster of spheres in the sticky-sphere limit. 
We define the sticky symmetry group of the cluster as the set of permutations and inversions of the spheres which preserve adjacency and can be realized by continuous deformations of the cluster that do not change the set of contacts or cause particles to overlap. The symmetry number is the size of the sticky symmetry group. We introduce a numerical algorithm to compute the sticky symmetry group and symmetry number, and show it works well on several test cases. 
Furthermore we show that once the sticky symmetry group has been calculated for indistinguishable spheres, the symmetry group for partially distinguishable spheres (those with non-identical interactions) can be efficiently obtained without repeating the laborious parts of the computations. 
We use our algorithm to calculate the partition functions of every possible connected cluster of 6 identical sticky spheres, generating data that may be used to design interactions between spheres so they self-assemble into a desired structure. 
\end{abstract}

\section{Introduction}

Symmetry plays an important role in the study of molecules and clusters of more general particles.  A group of particles, like atoms or colloids, can assemble into a variety of different clusters, and to determine which of these forms requires evaluating the number of symmetries of each cluster  \cite{sethna}. 
The definition of a symmetry varies depending on the  
particular system, quantities of interest, and method of calculating the partition function from statistical mechanics, but for a great many clusters -- those that are close to rigid, i.e. they don't move far from a reference configuration -- the definition coincides with the geometrical symmetries of the configuration, such as a symmetry upon reflection or upon rotation by certain angles around certain axes. Much effort has gone into developing methods to evaluate the geometrical symmetry group of molecules, an effort which is ongoing since this computation is a challenge for large molecules \cite{longuet, symmetry2, wales,symmetry5,sitharam}.

When a cluster is flexible, i.e. it can deform by some non-negligible amount along its internal degrees of freedom without a significant change in energy, the geometrical symmetry of the cluster is less meaningful, because the cluster almost always adopts a configuration with no or few true geometrical symmetries. However, the concept of symmetry is still meaningful, as long as one extends symmetry elements to include operations that continuously transform the cluster to one with a similar energy, called  ``feasible'' transformations in the seminal paper by Longuet-Higgins \cite{longuet}. 
This is a useful definition of symmetry because one can show that the number of such symmetries equals the number of times one has overcounted the cluster for the most common methods of evaluating its partition function. 
Calculating this kind of symmetry group requires more sophisticated methods since it requires understanding the kinetic pathways the cluster can follow, and not just its static geometrical symmetries. 
Some methods to compute such symmetry elements find pathways by searching a database of local minima and transition states on a cluster's energy landscapes, a process which can work well for atomic clusters, whose energy changes smoothly with the configuration \cite{wales,Wales:2014eh}.
Other methods are based on identifying a set of generators for the symmetry group, and then building the full group using group-theory software \cite{gap}. 
These methods often start with a particularly symmetric reference configuration with many geometical symmetries, or else identify symmetries by eye \cite{symmetry1}. 
Sometimes it is even possible to prove more general symmetry results about a cluster, and to extend these proofs to the more restrictive symmetry groups obtained when an energy function is present, for example as in \cite{flapan}.
Because identifying generators by eye or by first finding a particularly symmetric configuration, or proving statements about individual clusters, requires an external observer to provide input, 
these latter methods cannot be used to automatically evaluate the symmetries of a large collection of clusters. 

We are interested in particles with diameters of nano- to micrometres (\emph{colloids}), which are much larger than atoms, and form the building blocks for a wide range of materials \cite{Lu:2013dn,colloids}. 
Such particles interact attractively over scales typically much smaller than their diameters, so it is effective to model them in the \emph{sticky limit}, where the particles are treated as hard, classical rigid bodies that can't overlap, such as spheres, and the interaction potential is a delta function at the point of contact between a pair of bodies \cite{colloids,clusters}. Therefore, the energy of a cluster of particles is proportional to the number of pairs of particles that are exactly in contact. In contrast to atomic systems, this energy function changes abruptly at discrete locations in configuration space, so concepts developed for smoother energy landscapes, such as local minima and transition states, are no longer as meaningful. Therefore, algorithms for calculating the symmetry number of molecules with smooth interaction potentials will not work directly for these clusters. 

Our goal is to clarify the concept of symmetry for sticky-sphere clusters, possibly flexible, and to provide an algorithm that can evaluate the symmetry number of a sticky-sphere cluster automatically. 
We are interested both in indistinguishable particles and particles which can be partially distinguished. The theory developed here is a natural application of the theory of molecular symmetries from chemical physics, which is concerned with the types of symmetries present in a system's Hamiltonian. However,  this theory is usually presented in the context of quantum mechanics, and is not always straightforward to adapt to a purely classical setup \cite{cates}.  
 Our goal therefore is to present the theory of symmetry for sticky-sphere clusters in a mathematical framework, highlighting those aspects of the theory that depend only on geometry or topology and the overall connection to group theory. We hope this presentation will be accessible to those without a background in physics or chemistry, and will make it easy to adapt the theory to other situations where a modified definition of symmetry is needed. That the concept of symmetry for classical colloidal clusters needs clarification is evident in light of the numerous recent papers that attempt to explain it and the link to entropy in the statistical mechanics of such systems (e.g. 
\cite{Swendsen:2006gm,symmetry1,Frenkel:2014cn,cates}.)

Our theory starts with an equivalence relation between clusters, which says that two clusters are the same if one cluster can be continuously deformed into another without breaking any contacts or causing spheres to overlap. 
We define the sticky symmetry group of a sphere cluster to be the set of permutations or reflections of a cluster that can be achieved by such a continuous deformation, and the symmetry number to be the size of this group.  
This is a natural extension of the definition for rigid clusters, and we show how the sticky symmetry group is related to other symmetry groups commonly studied, the point group of the cluster and the automorphism group of its corresponding adjacency graph. 
Among the symmetry groups for flexible molecules studied in the literature, our approach is most closely related to the topological group studied by Flapan \cite{flapan}. The difference is that the symmetry group she considers arises from purely topological properties of the cluster's contact graph, while our symmetry group contains some geometry, since we require spheres to not overlap. 

We introduce a numerical algorithm to automatically compute the sticky symmetry group of a sphere cluster. 
The main component of the algorithm is a numerical method to find a continuous deformation linking one cluster to another, based on the steepest descent method in optimization \cite{kress} and simulated annealing \cite{liu}. 
 We apply this algorithm to several test cases and show it works well for small clusters. In addition we show that once the sticky symmetry group has been calculated for indistinguishable spheres, it can be efficiently computed when some groups of spheres are distinguishable from others, for any possible partitioning of the spheres into distinguishable groups. 
 
Ultimately, we wish to use the algorithm in an exhaustive computation of the partition functions of small clusters of sticky spheres (which depend on the symmetry numbers), and then ask how to design interactions between spheres so they self-assemble into a desired structure. Toward this aim, we compute the symmetry numbers and partition functions of every connected cluster of $N=6$ spheres, and comment on some of the physical insight this data gives us. 

The outline of the paper is as follows. In section \ref{sec:description} we give an overview of how the symmetry number arises when studying the statistical mechanics of sticky-sphere clusters. In Section \ref{sec:count} we develop the mathematical theory required to define a symmetry number for a flexible sticky-sphere cluster. In Section \ref{sec:num} we introduce a numerical algorithm to compute the symmetry number for a flexible cluster, and in Section \ref{sec:examples} we apply our theory and algorithm to several examples, including the exhaustive calculation for $N=6$ spheres.  Section \ref{sec:conclusion} concludes and discusses further applications in which this algorithm may be used.

\paragraph{Mathematical setup.} A cluster of $N$ spheres is a pair $(\bfx,\bfr)$, where $\bfx = (\bfx_1, \ldots, \bfx_N) \in \R^{3N}$ is the vector of sphere centers, with $\bfx_i \in \R^3$  the center of $i$-th sphere, and $\bfr = (r_1,\ldots,r_N)$ is the vector of sphere radii, with $r_i>0$ the radius of the $i$th sphere. 
We suppose that $m$ spheres are in contact, given by the set $E = \{ (i_1, j_1), \ldots, (i_m,j_m)\}$. 
When two spheres are in contact, their centers are related by the equation
\begin{equation}\label{eq:q}
 |\bfx_i-\bfx_j|^2 = (r_i+r_j)^2 , \quad (i,j) \in E\,.
\end{equation}
We assume that all non-contacting pairs of spheres do not overlap, so their centers must satisfy the inequality 
\begin{equation}\label{eq:bound}
|\bfx_i -\bfx_j|> r_i+r_j, \quad (i,j) \notin E\,.
\end{equation}
It is sometimes convenient to represent the set of contacts by the adjacency matrix $A$,  an $N \times N$ matrix whose entries are given by
\begin{equation}\label{eq:adj}
A_{ij} = \left\{ \begin{array}{cl}
1 & \text{if $(i,j) \in E$} \\
0 & \text{if $(i,j) \notin E$}
\end{array} \right.
\,.
\end{equation}
Notice that the pair $(\bfx,\bfr)$ is sufficient to characterize the cluster, since from it we can determine the set of contacts $E$ and hence the adjacency matrix $A$. 

We define $M_{A}^{(\bfr)}$ to be the set of all configurations $\bfx$ with adjacency matrix $A$ and radii $\bfr$, i.e. the set of points in $\R^{3N}$ which satisfy \eqref{eq:q}, \eqref{eq:bound}:
\begin{equation}\label{eq:m}
M_{A}^{(\bfr)} = \left\{ \bfy \in \R^{3N} : q_{ij}(\bfy) = (r_i+r_j)^2, \; \text{if} \; A_{ij} = 1, \; q_{ij}(\bfy) > (r_i+r_j)^2, \; \text{if} \; A_{ij} = 0 \right\}\,,
\end{equation}
where $q_{ij}(\bfx) =  |\bfx_i-\bfx_j|^2$. 
In the following, we will usually drop the dependence on the radii $\bfr$, and write $M_A$ for $M_{A}^{(\bfr)}$ and call a cluster $\bfx$, since $\bfr$ doesn't change for a given problem. 
 If the gradients $\{ \nabla q_{ij}(\bfy) \}_{(i,j)\in E}$ are linearly independent for every $\bfy \in M_A$, then $M_{A}$ is a manifold of dimension $p = 3n-m$.  
In general, $M_A$ is not connected, as we illustrate with examples in Section \ref{sec:examples}. We denote by $M_{A, \bfx}$ the connected component of $M_A$ to which cluster $\bfx$ belongs.  

The cluster $\bfx$ is \emph{rigid} if every point $\bfy$ in the connected component $M_{A, \bfx}$ can be obtained as a rotation or translation of $\bfx$, otherwise it is \emph{non-rigid} or \emph{flexible} \cite{rigid}. In the latter case, the cluster has internal degrees of freedom, other than translations and rotations: it can be continuously deformed without breaking contacts.

\section{Overview of the symmetry number and partition function for a sticky-sphere cluster}\label{sec:description}

\begin{figure}
\centering
\includegraphics[scale = 0.2]{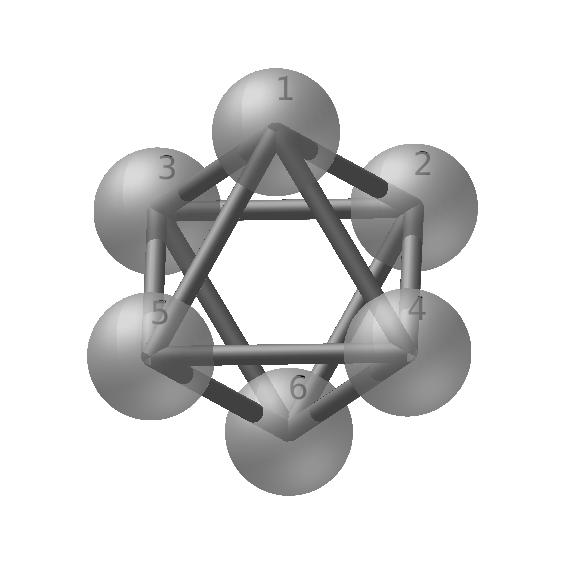}
\includegraphics[scale = 0.23]{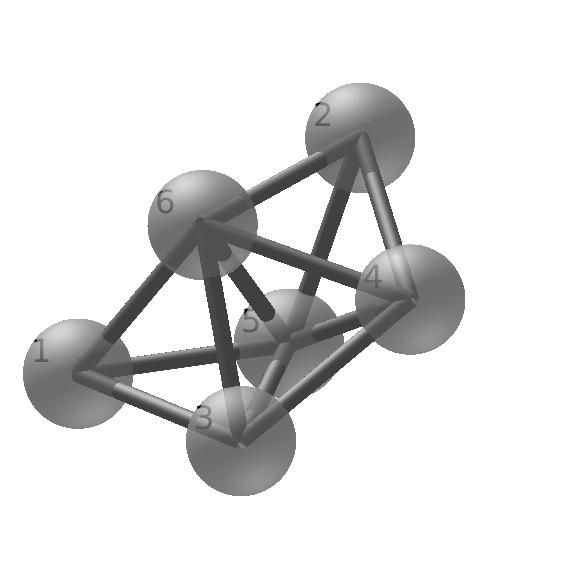}
\includegraphics[scale = 0.25]{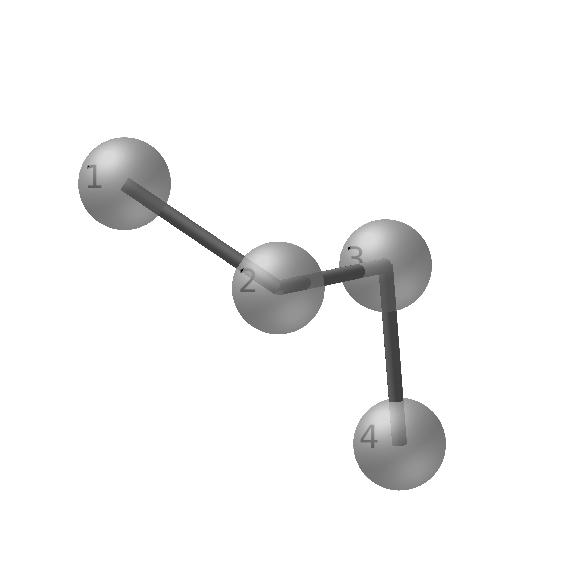}
\caption{Left and middle:  the two rigid clusters with $N = 6$ spheres, an octahedron (left) and a polytetrahedron (middle).  Each of these has several geometrical symmetriess. 
Right: a flexible chain of $N = 4$ spheres. In this configuration, the chain has no geometrical symmetries. 
All spheres are visualized with the radius half their actual size for clarity.}
\label{fig:loop4}
\end{figure}

We now give a brief overview of the partition function for a sphere cluster in the sticky-sphere limit, and  explain why the symmetry number enters. Readers interested in only in the mathematical definitions may skip to section \ref{sec:count}. 

We start with a cluster $\bfx$ of identical spheres with adjacency matrix $A$ with sticky interactions, and we wish to know the probability it will occur in equilibrium. We consider a system where we cannot tell $\bfx$ apart from a  cluster obtained by performing any of the following transformations: (i) rotating it, (ii) translating it,  (iii) deforming it along internal degrees of freedom without changing its adjacency matrix, (iv) reflecting it about an axis, or (v) permuting the particles. 
 In a classical system, such indistinguishability could arise because we cannot experimentally distinguish such clusters, or it could be imposed artificially even if we can distinguish them because we are only interested in properties of a group of clusters, and not individual clusters within the group. (In a quantum system, particles may be truly indistinguishable in a sense not possible in classical mechanics, but we do not consider quantum effects here.)
The equilibrium probability of finding $\bfx$ or any of its transformed versions is proportional to the partition function $Z_{\bfx}$, namely the integral of the Boltzmann distribution over the region in configuration space corresponding to all the possible transformations above of $\bfx$.\footnote{
Specifically, given the partition function for a collection of distinct (not related by any of the 5 transformations) clusters $\bfx\e{1},\bfx\e{2},\ldots,\bfx\e{k}$, the equilibrium probability to observe a cluster $\bfx\e{i}$ (or a transformation of it), given that at least one of the set above was observed, is $ Z_{\bfx\e{i}} / \sum_{j=1}^kZ_{\bfx\e{j}}$. 
} 
For sticky-sphere interactions where all pairwise interactions are identical, and when the constraints $\{ \nabla q_{ij}(\bfy) \}_{(i,j)\in E}$ are linearly independent\footnote{
When these constraints are not linearly independent, a similar expression may still be available but requires an additional parameter beyond $\kappa$, see \cite{Kallus:2017hi}.
} for every $\bfy \in M_{A, \bfx}$, the partition function can be written as \cite{miranda,Kallus:2017hi,clusters}
\begin{equation}\label{Zx}
Z_{\bfx} = C\: n_{\bfx}\: \kappa^m \int_{M_{A, \bfx}} f(\bfx') d\sigma(\bfx')\,.
\end{equation}
Here $\kappa>0$ is a system-dependent constant depending on the properties of the interaction potential and temperature, 
$f(\bfx)$ is a function depending only on the relative locations of the spheres (physically it is related to the vibrational entropy), 
 $d\sigma(\bfx)$ is the natural surface measure on $M_{A, \bfx}$ induced by restriction from the ambient Euclidian space, which also only depends on the relative locations of the spheres, and $C$ is a constant that is the same for all clusters with the same number of spheres $N$ (provided the space they live in is large enough that excluded volume effects do not matter.) 
We usually put additional restrictions on the cluster to make the integral \eqref{Zx} finite, for example by assuming the cluster is contained in a large box, or by fixing its center of mass. 

The remaining quantity in \eqref{Zx} is $n_{\bfx}$, which is the number of disconnected copies of $M_{A, \bfx}$ that we obtain by permuting or reflecting particles. 
Usually this factor is written as $2N!/\sigma$, where $\sigma$ is the so-called symmetry number of the cluster. Because the definition of symmetry is not fundamental, but depends on the method of calculating the partition function, we prefer to call $n_{\bfx}$ the \emph{counting number}, to highlight that it counts the number of geometrically isomorphic copies of $M_{A, \bfx}$ obtained by such transformations.

To see how the expression in \eqref{Zx} accounts for the five transformations listed above, notice that 
 $M_{A,\bfx}$ contains all possible transformations of $\bfx$ available through rotations, translations, and deformations along the internal degrees of freedom of a cluster. Indeed, these transformations don't alter the constraints in \eqref{eq:m} defining the manifold, and they are continuous so they take $\bfx$ to another point in $M_{A,\bfx}$. We only need to explicitly account for reflections and permutations, which can sometimes put $\bfx$ on a different connected component of $M_{A}$, or a different manifold $M_{A'}$ for $A\neq A'$. This accounting is done through the counting number. 

To see how, it is easiest to first consider a rigid cluster. 
An upper bound on the counting number is obtained by counting the total number of permutations and reflections: there are $N!$ ways of permuting the particles, and for each permutation we can reflect or not the cluster through some plane, so there are up to $2N!$ copies of the original cluster. 
This set of permutations and reflections are called \emph{permutation-inversion} (PI) operations \cite{longuet}. 
However, not all of these $2N!$ copies are distinct: some of these copies might correspond to a rotation or a reflection of the cluster, which we have already accounted for via the integral over $M_{A,\bfx}$.  
The number of such repeated copies is defined to be the \emph{symmetry number} $\sigma$ of the cluster, because each repeated copy corresponds to a geometrical symmetry element in the point group of the cluster \cite{wales}. 
The number of distinct copies of $M_{A,\bfx}$ available by permutations and reflections is then $n_{\bfx} = 2N!/\sigma$.

For clusters with the same number of contacts, the symmetry number can play an important role in determining the relative probabilities, as \eqref{Zx} shows that any ratio of symmetry numbers between clusters affects their relative probabilities by an equal ratio. 
Indeed, experiments have shown that in an ensemble of $6$ particles interacting with a short-range depletion interaction, the least symmetric cluster (the polytetrahedron, with symmetry $\sigma=4$, see Figure \ref{fig:loop4}) occurs 24 times more frequently than the most symmetric cluster (the octahedron, with $\sigma=48$, see Figure \ref{fig:loop4}), and most of this discrepancy is due to the ratio of symmetry numbers \cite{Meng:2010gsa}. 

For a flexible cluster, the counting number must account for the fact that the cluster has internal degrees of freedom. Therefore, even if none of the $2N!$ permutation-inversion operations correspond to a perfect geometrical symmetry of $\bfx$, they may still transform the cluster so it lies on the same connected component $M_{A,\bfx}$, and therefore has already been accounted for in the integral over $M_{A,\bfx}$. 
For example, consider a chain of four identical spherical particles, labelled 1--4, as in Figure \ref{fig:loop4}. This particular embedding $\bfx$ of the chain is purposely chosen so the chain has no geometrical symmetries; its point group is trivial. However, the permutation $(14)(23)$, which reverses the ordering of the spheres, results in a cluster that can be obtained from the original one by a continuous deformation, for example by straightening the chain, rotating it by 180$^\circ$, then crinkling it up again. The permuted cluster therefore lies on the same connected component $M_{A,\bfx}$.
In this case, to integrate over the correct space of transformations if we start with the integral expression in \eqref{Zx},  we need to define the symmetry number $\sigma$ to be the number of permutation-inversion operations that can also be realized by a combination of transformations (i--iii), namely translating, rotating, and internal deformations; none of these need correspond to an actual geometrical symmetry. 
The counting number is still the same, $n_{\bfx}=2N!/\sigma$. 

These ideas may be easily extended to the case of partially distinguishable spheres, i.e. where certain groups of spheres are distinguishable. 
Physically, this situation may arise when the particles have different kinds of interactions with each other, or are physically different in some way such as in the materials they are made with. Such situations are of interest in self-assembly problems where one wants to design interactions between particles to so they assemble into a desired structure \cite{hormoz,Zeravcic:2014it}. 
For example, we may consider spheres in some set $C_1 = \{1,2\}$ to be indistinguishable from each other, and those in set $C_2=\{3,4,\ldots,N\}$ to be indistinguishable from each other, but we can determine whether a sphere belongs to set $C_1$ or $C_2$. 
It is convenient to think of the spheres as having different colors, so 
 we may call spheres in $C_1$ ``blue'' and spheres in $C_2$ ``red.'' In general there could be anywhere from $1$ to $N$ colors. 

When there are two or more colors, what changes in the partition function \eqref{Zx} are the following: (i) the factor $\kappa^m$ is replaced by $\prod_{(i,j)\in E} \kappa_{ij}$, a product of sticky parameters $\kappa_{ij}$ for each pair $(i,j)$ in contact in the cluster; and (ii) the counting number $n_{\bfx}$ must be recomputed to account for the number of different copies of $M_{A,\bfx}$ one obtains under the more restrictive set of permutations which preserve the colors. Notably, the integral over $M_{A,\bfx}$, which is laborious to compute \cite{mcmc},  does not change. An important contribution of this paper will be to show that the counting number for colored spheres can be efficiently obtained from the counting number for indistinguishable spheres. Therefore, using this framework one can easily compute the partition function after changing the interactions (sticky parameters) and colorings of the spheres, once the integral and counting number have been computed for indistinguishable spheres.

\section{Counting number, symmetry number, and sticky symmetry group}\label{sec:count}

This section more precisely elaborates on the ideas introduced in section \ref{sec:description}. 
Our aim is to define the counting number and symmetry number of a cluster and show how they are related to the number of isometric copies of manifolds $M_{A,\bfx}$. 

\medskip

We start by fixing the radii $\bfr = (r_1, \ldots, r_N)$ and defining the set $Y$  of all the physical realizations of sticky-sphere clusters with that set of radii:
\begin{equation*}
Y= \R^{3N} \setminus \{ \bfx \in \R^{3N} : |\bfx_i-\bfx_j| < r_i+r_j, \; \exists i, j, \; i\neq j \}.
\end{equation*}
Within set $Y$ the adjacency matrix of a cluster  $A(\bfx)=A(\bfx,\bfr)$ is well-defined. 

We introduce in $Y$ an equivalence relation $\sim$ between clusters, that tells us which clusters are assumed to be ``the same,''  written $\bfx\sim \bfy$. We define
\begin{equation}\label{eq:sim}
\bfx \sim \bfy \;\;\Leftrightarrow\;\; A(\bfx) = A(\bfy), \text{ and }
\exists \varphi:[0,1]\to M_A, \varphi \text{ cts}, \text{ s.t. }\varphi(0)=\bfx,\; \varphi(1)=\bfy.
\end{equation} 
It is trivial to verify that $\sim$ is an equivalence relation. 

In words, $\bfx\sim \bfy$ if these clusters have the same adjacency matrix, and there is a continuous deformation, formed from some combination of rotations, translations, and motion along internal degrees of freedom, from $\bfx$ to $\bfy$. 
By construction, if $\bfx\sim\bfy$ then they belong to the same connected component of $M_{A(\bfx)}$, i.e. $M_{A(\bfx),\bfx} = M_{A(\bfy), \bfy}$. 

The set of clusters which can be obtained by a continuous deformation of $\bfx$ is the equivalence class of $\bfx$, written $[\bfx]$. A trivial but important relation is that
\begin{equation}\label{relation}
[\bfx]=M_{A,\bfx}\,.
\end{equation}
Two clusters are distinct if they belong to different equivalence classes: they either have different adjacency matrices, or there is no continuous deformation from one to the other. 

We denote by $X$ the quotient set of $Y$ modulo $\sim$:
\begin{equation}\label{eq:Xr}
X = Y/\sim.
\end{equation} 
$X$ is the set of all distinct equivalence classes of clusters. 

We point out that each of $Y$, $\sim$, $X$ depends on $\bfr$, but we suppress this dependence in the notation for brevity.

We also need to consider the set of permutations and reflections of a cluster $\bfx$.
A reflection of $\bfx$ is simply $-\bfx$; any reflection about any other plane is then obtained by rotating a particular reflection. 
To construct a permutation of $\bfx$ we start with a permutation matrix $P$, an $N \times N$ matrix with entries $P_{ij}=1$ if $i\to j$ after the permutation is applied, and $P_{ij}=0$ otherwise.  The permuted cluster is then $(P \otimes I_3) \bfx$, where $I_3$ is the $3 \times 3$ identity matrix, and $\otimes$ is the Kronecker product of matrices  \cite{jones}.  
For brevity we write $\widetilde{P}=P \otimes I_3$ so that $\widetilde{P}\bfx$ is the permuted version of $\bfx$. 

Here are two useful facts about a permutation matrix $P$: one,  $P^{-1} = P^T$, since $P$ is an orthogonal matrix, and two, if  the adjacency matrix for $\bfx$ is $A$, then the adjacency matrix for $\widetilde{P}\bfx$ is $PAP^T$. 

To make the link to symmetry groups, we need to consider how the group of permutation-inversion operations acts on $[\bfx]$. 
It is simplest to first consider spheres with identical radii that are indistinguishable; we do this in Section \ref{sec:identical}. Then, we consider extensions to partially distinguishable spheres (Section \ref{sec:col}) and spheres with different radii (Section \ref{sec:rad}).

\subsection{Indistinguishable spheres with identical radii}\label{sec:identical}

Throughout this section we assume $\bfx$ is a cluster of $N$  indistinguishable spheres with identical radii. 

\paragraph{Basic results from group theory.}
Consider the product group 
\begin{equation}\label{eq:prod_group}
P(N) \times C_2 = \{ (P,\delta) : P\in P(N), \; \delta \in C_2 \},
\end{equation}
where $P(N)$ is the group of $N \times N$ permutation matrices, and $C_2 = \{ \pm 1\}$. An element in $P(N) \times C_2$ is called a \emph{permutation inversion} (PI) operation \cite{longuet}.  

We define an action of $P(N) \times C_2$ on the quotient set $X$ given in \eqref{eq:Xr} as 
\begin{equation}\label{eq:action}
(P,\delta) \cdot [\bfx] \; = \; 
 [\delta\widetilde{P}\bfx] \,.
\end{equation}
That is, the action permutes the spheres in $\bfx$ and possibly reflects it. 
The proof that $\cdot$ is a well-defined group action is given in the Appendix.

We define the \emph{counting number} $n_{\bfx}$ of a cluster $\bfx \in \R^{3N}$ to be the size of the orbit of $[\bfx]$ with respect to the action $\cdot$, written $\mbox{orb}([\bfx])$:
\begin{equation}\label{eq:sym_num}
n_{\bfx} = |\mbox{orb}([\bfx])| =  \left|\{ [\delta \widetilde{P}\bfx ] : P \in P(N), \delta \in C_2 \}\right|.
\end{equation}
In words, the counting number equals
the number of distinct copies of $[\bfx]$ that we obtain by permuting spheres in $\bfx$ or reflecting it. 

Some permutations or inversions of $\bfx$ leave its equivalence class unchanged. 
The set of such operations forms the stabilizer of $[\bfx]$, $\mbox{stab}([\bfx])$: 
\begin{equation}\label{eq:stab1}
\mbox{stab}([\bfx]) = \{ (P,\delta) \in P(N) \times C_2 : (P,\delta) \cdot [\bfx] = [\bfx]\}.
\end{equation}
It is a fact of group theory that $\mbox{stab}([\bfx])$ is a subgroup $P(N)\times C_2$ \cite{jones}. 

The orbit-stabilizer theorem relates the sizes of the orbit and stabilizer of $[\bfx]$ \cite{jones}: 
\begin{equation}\label{OSthm}
|\mbox{orb}([\bfx])| = \frac{|P(N) \times C_2|}{|\mbox{stab}([\bfx])|}\,.
\end{equation}

\paragraph{Sticky symmetry group and relation to manifolds $M_{A(\bfx),\bfx}$.}

We will now show how this theorem is related to the number of distinct copies of $M_{A(\bfx),\bfx}$, and consider a more explicit way to define the stabilizer. 

Because of \eqref{relation}, 
the group action \eqref{eq:action} can also be thought of as acting on manifolds as
\begin{equation}\label{actionM}
(P,\delta) \cdot M_{A(\bfx),\bfx} = M_{PA(\bfx)P^T, \delta\widetilde{P}\bfx}\,.
\end{equation}
Therefore, an element $(P,\delta)$ belongs to the stabilizer if and only if $PA(\bfx)P^T = A(\bfx)$,  and both $\bfx$ and $\delta \widetilde{P}\bfx$ belong to the same connected component of the manifold $M_{A(\bfx)}$.
This relation gives another way to define the stabilizer that will prove useful in computations. 

Recall that  the \emph{automorphism group}  of an adjacency matrix $A$ is the set of all permutation matrices $P$ that preserve adjacency:
\begin{equation}\label{eq:aut}
\G = \text{Aut}(A) = \{ P \in P(N) : PA = AP \}.
\end{equation}
An element $P\in \G$ corresponds to a permutation that, when applied to spheres in a cluster, doesn't change who each sphere is in contact with. 
Notice that the automorphism group is independent of the embedding of  $\bfx \in \R^{3N}$ and the radii $\bfr$; it is a property only of the graph associated with the adjacency matrix. 

Using this definition we can write the stabilizer more explicitly as:
\begin{equation}\label{eq:top_group}
\T_{\bfx} := \mbox{stab}([\bfx]) = \{ (P, \delta) \in \G \times C_2  : \exists \varphi : [0,1] \rightarrow M_A, \varphi \text{ cts}, : \varphi(0) = \bfx, \; \varphi(1) = \delta \widetilde{P} \bfx \}\,.
\end{equation}
We call $\T_{\bfx}$ the  \emph{sticky symmetry group} of the cluster. 
From this definition it is clear that $\T_{\bfx}$ is a subgroup of $\G \times C_2$, a group we call the \emph{automorphism-inversion group}.
In particular, $\T_{\bfx}$ is the set of the elements of the automorphism-inversion group $\G$ that can be obtained as either a continuous deformation of the cluster, or a deformation combined with a reflection. Notice that $\T_{\bfx}$ depends on the class $[\bfx]$, in contrast to $\G\times C_2$, which doesn't. 

We define the  \emph{symmetry number} $\sigma_{\bfx}$ of the cluster $\bfx$ to be the cardinality of $\T_{\bfx}$:
\begin{equation}\label{eq:symmetrynumber}
\sigma_{\bfx} = |\T_{\bfx}|.
\end{equation}
Combined with the orbit-stabilizer theorem \eqref{OSthm}, we obtain the relationship 
\begin{equation}\label{eq:nx}
n_{\bfx} = \frac{2N!}{\sigma_{\bfx}}.
\end{equation}

Because of \eqref{relation} and \eqref{actionM},  $n_{\bfx}$ equals the number of disconnected manifolds one obtains by applying all the permutation-inversion operations to $M_{A(\bfx),\bfx}$. 
Why is it reasonable for $n_{\bfx}$ to appear in the partition function \eqref{Zx}?
We show in the Appendix that  the mapping $\bfx\to \delta \widetilde{P}\bfx$ is an isometry, and therefore 
$\int_{M_{A, \bfx}} f(\bfy) d\sigma(\bfy) = \int_{M_{PA(\bfx)P^T, \delta\widetilde{P}\bfx}}f(\bfy) d\sigma(\bfy)$. 
Therefore, this factor accounts for the integral over the parts of configuration space that we wish to include in the partition function, but that are identical to the factor already computed in the integral over $M_{A(\bfx),\bfx}$. 


\medskip

We point out that $\T_{\bfx}$ provides information on the connectivity of the manifold $M_A$. In particular, if $\T_{\bfx} \subsetneq \G \times C_2$, then $M_A$ is disconnected, with certain components that are not related by reflections.\footnote{
If a cluster is chiral, i.e. there is no continuous transformation between $\bfx$ and $-\bfx$, then $M_A$ will be disconnected simply because $M_{A(\bfx),\bfx}\neq M_{A(\bfx),-\bfx}$. 
}
 To see why, suppose $(P,\delta) \in \G \times C_2 \setminus \T_{\bfx}$ is an element in the automorphism-inversion group but not in the sticky symmetry group, and let $\bfy = \delta \widetilde{P}\bfx$. 
 Since $P \in \G$, $\bfx$ and $\bfy$ belong to the same manifold $M_A$. However, since $(P,\delta) \notin \T_{\bfx}$, there exists no path in $M_A$ connecting $\bfx$ with $\bfy$, and therefore $\bfx$ and $\bfy$ belong to different connected components of $M_A$. 
We will provide a concrete example where this happens in Section \ref{sec:examples}.

The converse is not true in general: $M_A$ could be disconnected even if $|\T_{\bfx}|=|\G|$ (note we do not say $\T_{\bfx}=\G$ because each element of $\T_{\bfx}$ is associated with an inversion while the elements in $\G$ are not. However, in most cases only one of $+P,-P$ is in $\T_{\bfx}$.) We will provide a counterexample in section \ref{sec:examples}.

\paragraph{Relationship to the point group.}

It is useful to relate the sticky symmetry group to another group widely studied in physics and chemistry, the point group of a cluster. The \emph{point group} $\Po_{\bfx}$ of the cluster $\bfx$ is the set of all the elements of the automorphism-inversion group that can be realized as a rotation of $\bfx$ \cite{wales}:
\begin{equation}\label{eq:PG}
\Po_{\bfx} = \{ (P,\delta) \in \G \times C_2 : \exists R \in SO(3): \delta \widetilde{P}\bfx = (R \otimes I_N) \bfx   \}\,.
\end{equation}

There is a useful relation between the point group, the sticky symmetry group, and the automorphism-inversion group of a cluster:
\begin{equation}\label{eq:incl}
\Po_{\bfx} \subseteq \T_{\bfx} \subseteq \G \times C_2.
\end{equation}
That is, the sticky symmetry group $\T_{\bfx}$ is a subgroup of the automorphism-inversion group, and it contains the point group $\Po_{\bfx}$. 

To see that $\Po_{\bfx} \subseteq \T_{\bfx}$, let $(P,\delta) \in \Po_{\bfx}$. By the definition of point group  \eqref{eq:PG}, there exists an orthogonal matrix $R \in SO(3)$ such that $(R \otimes I_N)\bfx = \delta \widetilde{P}\bfx$. We now show we can achieve this transformation continuously.  Since $SO(3)$ is connected, there exists a continuous path $R(t) : [0,1] \rightarrow SO(3)$ such that $R(0) = I_3$ and $R(1) = R$. We define the path $\sigma(t) = (R(t) \otimes I_N)\bfx$. For every $t \in [0,1]$, $\sigma(t) \in M_A$ since 
\begin{equation*}
q_{ij}(\sigma(t)) = |R(t)\bfx_i-R(t)\bfx_j|^2 = |R(t)(\bfx_i-\bfx_j)|^2 = |\bfx_i-\bfx_j|^2 = q_{ij}(\bfx),
\end{equation*}
and therefore $q_{ij}(\sigma(t)) = (r_i+r_j)^2$ if $A_{ij} = 1$, and $q_{ij}(\sigma(t))> (r_i+r_j)^2$ if $A_{ij} = 0$. Moreover, $\sigma(t)$ is continuous and $\sigma(0) = \bfx$, $\sigma(1) = (R \otimes I_N)\bfx = \delta \widetilde{P}\bfx$. Therefore, $(P,\delta)$ belongs to the sticky symmetry group $\T_{\bfx}$.

Relationship \eqref{eq:incl} is useful for computing the sticky symmetry group of clusters that are not too large, because for such clusters the automorphism and point groups may be calculated on reasonable timescales.  
The automorphism group may be computed using algorithms from graph theory \cite{nauty}, and 
the automorphism-inversion group is obtained as a direct product. 
Note that calculating the automorphism group becomes a challenge
 for clusters that are not small, as its size can grow extremely rapidly with $N$ (e.g. see examples in \cite{symmetry3}.) 

If one can compute the automorphism-inversion group, then one can obtain the 
 point group $\Po_{\bfx}$ by checking which automorphisms also preserve the set of pairwise distances \cite{graph1}. Let $D_{\bfx}$ be the $N \times N$ matrix whose entries measure the squared distance between the centers of the spheres:
\begin{equation}\label{eq:dist_mat}
(D_{\bfx})_{ij} = |\bfx_i - \bfx_j|^2\,.
\end{equation}
Then  $(P,\delta) \in \Po_{\bfx}$ for some $\delta$ if and only if $PD_{\bfx}P^T = D_{\bfx}$; in other words, all the pairwise distances are the same after applying the permutation. (Of course, to compute the point group one also needs to determine $\delta$, which can be done straightforwardly). 
Although this result is widely used in physics and chemistry to compute the point group of molecules, we have not found a rigorous proof of it in the literature (we found the standard reference \cite{YH38} to be incomplete), and therefore include one in the Appendix for completeness.

\medskip

\noindent \emph{Remark.} In physics and chemistry, the counting number $n$ of a rigid molecule with $N$ atoms is often computed as
\begin{equation*}
n = \frac{\xi N!}{\sigma},
\end{equation*}
where $\xi$ is 2 if the molecule is chiral and 1 otherwise, and $\sigma$ is the number of \emph{rotations} which are equivalent to a permutation of the atoms. 
Recall that a cluster $\bfx$ is chiral if there is no rotation-translation operation which maps the cluster to its reflection $-\bfx$. 
This is consistent with our formula \eqref{eq:nx}: if the cluster is chiral, then $\xi = 2$ and our symmetry number is the same as the one above, $\sigma_{\bfx} = \sigma$, since no reflections belong to the point group in this case. If the cluster is achiral, then $\xi = 1$ and our symmetry number is $\sigma_{\bfx}=2\sigma$, since half of the elements of the point group are rotations. 

\medskip

\noindent \emph{Remark.}
In many applications one wishes to distinguish reflections of a cluster. This framework can be adapted to such a situation, by removing the outer product with $C_2$ in all the groups under consideration: the sticky symmetry group $\mathcal T_{\bfx}^0$ would be the set of permutations that can be achieved by a continuous deformation, and the counting number would be $|\mathcal G| / |\mathcal T^0_{\bfx}| = N! / |\mathcal T^0_{\bfx}|$.


\subsection{Colored particles}\label{sec:col}

The theory developed so far can easily be adapted to colored (partially distinguishable) particles. 
Let $\bfx$ be a cluster with adjacency matrix $A$, and let $V = \{1,2,\ldots,N\}$ be labels, one for each sphere. 
We still assume the spheres have identical radii. 
Suppose $V$ is partitioned into $k$ disjoint subsets $\C = \{ C_i \}_{i=1}^k$, where $k$ is the number of colors, so that $V = \bigcup_{i=1}^k C_i$. 
We denote by $\bfx^{\C}$ the cluster colored according to the partition $\C$. 

A permutation $P \in P(N)$ acts on the partition as 
\begin{equation}\label{eq:part_p}
P \cdot C_i = \{ \pi(j_1^{(i)}), \ldots, \pi(j_{p_i}^{(i)}) \}, \qquad i = 1, \ldots, k,
\end{equation}
where $\pi \in S_N$ is the permutation associated with $P$, and  $C_i = \{ j^{(i)}_{1}, \ldots, j_{p_i}^{(i)} \}$. 
To count the number of distinct copies of $M_{A,\bfx}$ that we obtain by permuting the particles, we now consider only those permutations which preserve the partition $\C$.  
In particular, we define the group 
\begin{equation}\label{eq:pnc}
P(N)^{\C} = \{ P \in P(N) : P \cdot C_i = C_i, \; \forall i=1, \ldots, k\},
\end{equation}
which is a subgroup of $P(N)$, and consider the product group $P(N)^C \times C_2$.

The sticky symmetry group $\T_{\bfx^\C}$ of the colored cluster $\bfx^\C$ is\begin{equation}\label{eq:top_col}
\T_{\bfx^\C} = \{ (P,\delta) \in \T_{\bfx} : P \cdot C_i = C_i, \; \forall i = 1, \ldots, k \}.
\end{equation}
That is, it is those elements of  the sticky symmetry group of the cluster when the spheres are indistinguishable, $\T_{\bfx}$ (see \eqref{eq:top_group}), that preserve the partition. 

  The counting number $n_{\bfx^\C}$ for the colored cluster is then given by 
\begin{equation}
n_{\bfx^\C} = \frac{2|P(N)^\C|}{\sigma_{\bfx^\C}},
\end{equation}
where $\sigma_{\bfx^\C} = |\T_{\bfx^\C}|$. 

It is a useful fact that once the sticky symmetry group has been computed for indistinguishable particles, with $k=1$, we may easily obtain the sticky symmetry group and hence symmetry number for any coloring, simply by checking which elements of $\T_{\bfx}$ preserve the particular color labels.

\subsection{Particles with different radii}\label{sec:rad}

The framework above may also be adapted to spheres with different radii. 
Specifically,  let $(\bfx, \bfr)$ be a cluster in $\R^{3N}$ with radii $\bfr = (r_1, \ldots, r_N)$, of which there are $n_r$ distinct radii. 
There is a natural partition $\mathcal R = \{ R_i \}_{i=1}^{n_r}$ of the spheres $V=\{1,2,\ldots,N\}$ obtained by grouping together spheres with the same radius.  
We may then proceed as in the case of colored particles, substituting $\mathcal R$ for the partition $\C$. 
The only significant change occurs when computing the sticky symmetry group of particles with different radii, since the manifold of configurations $M_A^{(\bfr)}$ depends on the radii (see \eqref{eq:m}). 
Contrary to the case of colored particles, changing the radii $\bfr$ and partition $\mathcal R$ implies a new computation of the sticky symmetry group $\T_{\bfx^\mathcal R}$. 

Finally, if a cluster has particles with different colors and radii, we first determine the sticky symmetry group $\T_{\bfx^\mathcal R}$ according to the partition $\mathcal R$, then check which elements of $\T_{\bfx^\mathcal R}$ preserve the partition into colors $\mathcal C$ as in \eqref{eq:top_col}. We have to make sure the colors respect the radii partition, i.e. two spheres with different radii cannot have the same color.

\section{Numerical algorithm}\label{sec:num}

In this section we provide a numerical algorithm to compute the sticky symmetry group of a cluster $\bfx$, and hence its symmetry number. The algorithm first computes the automorphism and point groups of $\bfx$, and then it checks each automorphism $P$ which is not in the point group to see if either of $(P,+1)$ or $(P,-1)$ are in the sticky symmetry group. To do this, we search for a continuous path in $M_A$ connecting $\bfx$ with $\widetilde{P}\bfx$ and with $-\widetilde{P}\bfx$. If we find a path for either of these cases, we add $(P,\delta)$ to the sticky symmetry group. 
Although there is no guarantee that we will find a path if there is one, we show that the algorithm performs very well on several test cases. 

Here is a summary of our algorithm. 
\begin{enumerate}
\item We compute the automorphism group $\G$ of the cluster (see \eqref{eq:aut}). 
This can be done using graph-theoretical algorithms such as Nauty \cite{nauty}. 
For very small cases one can simply check each of the $N!$ permutations to see if it preserves adjacency. 
\item We compute the point group $\Po_{\bfx}$: we calculate the distance matrix $D_{\bfx}$ of $\bfx$ as in \eqref{eq:dist_mat}, and check which elements in $\G$ preserve the distance matrix (see \eqref{eq:point_2}). We then determine $\delta$. In our implementation we found $\delta$ using the path algorithm specified in Section \ref{sec:paths} below, simply because it was convenient, but there are simpler methods. 
\item We compute the sticky symmetry group $\T_{\bfx}$. 
Since it contains $\Po_{\bfx}$ (see \eqref{eq:incl}), we must only check each element in $\G \times C_2 \setminus \Po_{\bfx}$.  For every element $(P,\delta)$ in this set, we check if there exists a continuous path in the manifold $M_A$ connecting $\bfx$ and $\delta \widetilde{P}\bfx$, using the algorithm specified in Section \ref{sec:paths} below. If such a path exists, we add $(P,\delta)$ to the sticky symmetry group $\T_{\bfx}$. 
\item If the particles are colored with partition $\C$, we compute the sticky symmetry group $\T_{\bfx^\mathcal C}$ as in \eqref{eq:top_col}, by checking which elements of $\T_{\bfx}$ preserve the partition. 
\end{enumerate}

We remark that we can use the fact that $\T_{\bfx}$ is a group to check for errors in our algorithm. If for some $(P,\delta)$ we fail to find a continuous path from $\bfx$ to $\delta \widetilde{P} \bfx$, but one exists, we can sometimes find it afterwards by checking to see if our numerically-computed $\T_{\bfx}$ is a group, for example using group theoretical algorithms such as GAP \cite{gap}. If the elements we have found do not form a group, the algorithm computes the smallest group containing these found elements. 
We found this check particularly useful when developing our path-finding algorithm, where we did sometimes fail to find certain path, but in the path-finding algorithm's current state we haven't found examples where it fails. 
We expect this check to be useful if the error rate is small, as then the check can fill in the rare missing elements. 

We further remark that in computing the sticky symmetry group in step 3  above, it is actually sufficient to find only the generators of the group $\T_{\bfx}$. Specifically, one could just find the elements which, together with $\Po_{\bfx}$, generate (as a group) the sticky symmetry group $\T_{\bfx}$, and then generate the entire group using a group theoretical algorithm. We are not aware, though, of a method to test if an exhaustive set of generators of $\T_{\bfx}$ have been found.

The major part of the algorithm is looking for paths in the manifold $M_A$. We describe a numerical algorithm in the next paragraph. This algorithm may be applied more generally to manifolds defined by equality and inequality constraints, as we will discuss.

\subsection{Finding continuous deformation paths}\label{sec:paths}

We consider a slightly more general setup than that in section \ref{sec:description}. 
Let $M$ be a set in $\R^d$ implicitly defined by 
\begin{equation}\label{eq:md}
M = \{ \bfx \in \R^d : q_i(\bfx) = 0, \; i = 1, \ldots, m, \; h_j(\bfx) > 0, \; j = 1, \ldots, l \},
\end{equation}
where $q_i, h_j : \R^d \rightarrow \R$ are smooth functions. We assume the gradients $\{ \nabla q_i(\bfx) \}_{i=1}^m$ are linearly independent for every $\bfx \in M$, which implies that $M$ is a differentiable manifold of dimension $p= d-m$. Let $\bfx_0,\bfx_1 \in M$ be two distinct points in $M$. Our aim is to determine if $\bfx_0$ and $\bfx_1$ belong to the same connected component of $M$, i.e. if there exists a continuous path $\varphi : [0,1] \rightarrow M$ such that $\varphi(0) = \bfx_0$ and $\varphi(1) = \bfx_1$. 

We start with the following observation. Let's forget about being on a manifold, and suppose we want to minimize in $\R^{d}$ the function 
\begin{equation*}
U(\bfy) = |\bfy - \bfx_1|^2
\end{equation*}
using the method of steepest descent starting from point $\bfx_0$ \cite{kress}. We do this by computing a sequence of points $\bfy_0,  \bfy_1,\ldots $ with $\bfy_0 = \bfx_0$ and 
\begin{equation}\label{seq_sd}
\bfy_{k+1} = \bfy_k - t_k \nabla U(\bfy_k), \quad k = 0, 1, \ldots\,.
\end{equation}
The best choice of  $t_k$ is the one that minimizes the function 
\begin{equation*}
\phi_k(t) = U(\bfy_k-t\nabla U(\bfy_k)) = (1-2t)^2|\bfy_k-\bfx_1|^2.
\end{equation*}
The minimum occurs at $t_k=1/2$, 
which implies that the second point in the sequence is 
$\bfy_{1} = \bfy_0 - \frac{1}{2} \cdot 2 (\bfy_0 -\bfx_1) = \bfx_1$. 
The steepest descent in this case is trivial, leading us in one step to the true minimum. 

Suppose now that we want the sequence \eqref{seq_sd} to approximate a continuous path and to be constrained to the manifold $M$. To approximate a continuous path we may put an upper bound on the step size, say \texttt{tol}. To be constrained to the manifold, we project the steepest descent direction $\bfx_1-\bfy_k$ to the tangent space to the manifold, take a step in this direction, and then project back to the manifold. We choose a projection such that the correction step is perpendicular to the tangent space at $\bfy_k$, so that if there were no limit on the step size, the complete step including tangent step plus projection step would take us directly to $\bfx_1$.  

Specifically,  suppose $\bfy_k \in M$ is the point generated at step $k$, and let $T_kM$ denote the tangent space of $M$ at $\bfy_k$ and $T_k^\perp M$ its complement, the normal space. Let $P_k : \R^d \rightarrow T_kM$ be the matrix which orthogonally projects a vector to the tangent space. 
We generate a sequence of points as 
\begin{equation}\label{eq:yk}
\bfy_{k+1} = \bfy_k + \Delta s\: \bfu_k  + \bfw_k, \qquad k = 0,1, \ldots, 
\end{equation}
where $\Delta s >0$ is the step-size, $\bfu_k\in T_kM$ is a unit vector in the tangent space at $y_k$, and $\mathbf{w}_k \in T_k^\perp M$ is a vector in the normal space at $y_k$. 
The direction in the tangent space is chosen by projecting the steepest descent direction as 
\begin{equation}\label{eq:unit}
\bfu_k = \frac{P_k(\bfx_1-\bfy_k)}{|P_k(\bfx_1-\bfy_k)|}.
\end{equation}
The vector $\bfw_k$ in the normal space is chosen so that $\bfy_{k+1}\in M$. 
Specifically, since $\{\nabla q_{ij}(\bfy_k)\}_{(i,j)\in E}$ spans $T^\perp_k$, we let 
$\bfw = \sum_{j=1}^m a_j \nabla q_j(\bfy_k)$ for some unknown coefficients $\{a_j\}_{j=1}^m$, and then 
solve the system of equations $q_i\left(\bfz+\sum_{j=1}^m a_j \nabla q_j(\bfy_k)\right) = 0$, $i=1, \ldots, m$ using Newton's method. Generically the solutions for $\bfw_k$ will be isolated. 

The step size is chosen as
\begin{equation}\label{eq:ds}
\Delta s = \min \{ \texttt{tol},  |P_k(\bfx_1-\bfy_k)|\}. 
\end{equation}
This choice is motivated by the following observation: if $\Delta s = |P_k(\bfx_1-\bfy_k)|$, then we would find a solution for the normal step is $\bfw_k = P_k^\perp (\bfx_1-\bfy_k)$, where $P_k^\perp$ is the orthogonal projection matrix to $T_k^\perp M$, so $\bfy_{k+1} = \bfx_1$: the algorithm would bring us to the optimal point in one step. 
The step in the optimal direction is simply broken up into a step in the tangent space, and a step in the normal space. We impose an upper bound \texttt{tol} on the step size to approximate a continuous path, but take a smaller step if it would be better.

If ever  $|\bfy_k-\bfx_1|<\texttt{tol}$ the algorithm would take us to $\bfx_1$ in one step, so we stop generating points, and say we have found a path from $\bfx_0$ to $\bfx_1$. 

When the manifold has a boundary $\partial M$, because there are one or more inequalities in \eqref{eq:md}, the sequence $\bfy_k$ may fall outside of the boundary. Therefore, we check at each step if $\bfy_k$ satisfies the inequalities, i.e. if $h_j(\bfy_k) > 0$, for every $j = 1, \ldots, l$. If not, we stop the steepest descent and switch to a different, random method to generate points. 

We switch to random steps if any of the following occur: 
\begin{enumerate}[(a),nosep]
\item A point $\bfy_k$ lies outside the boundary. 
\item The projection onto the manifold fails, i.e. we fail to find a solution $\bfw_k$ to the nonlinear equations. 
\item The point $\bfy_k$ is stuck in a local mininum, i.e. $|\nabla q_{ij} \cdot \bfy_k| < \texttt{tolN}$, for some $i,j = 1, \ldots, N$, and $\texttt{tolN}$ is a numerical tolerance. 
\end{enumerate}
\medskip

If we decide to switch to random steps, we generate some number $\texttt{Nr}$ of random steps, and then resume the steepest descent. 
We tried two different methods for generating random points. 

\begin{enumerate}
\item \emph{Random steps}. Given a point $\bfy_k \in M$, we generate a random vector $\bfv \in T_{k}M$ according to the isotropic Gaussian density
\begin{equation}\label{eq:dens}
p(\bfv) = \frac{1}{(2\pi)^{p/2}\sigma^p}\exp\left(-\frac{|\bfv|^2}{2\sigma^2}\right),
\end{equation} 
where $\sigma >0$ is a parameter and $p$ is the dimension of the manifold $M$. We then project the point  $\bfy_k+\bfv$ back to the manifold using the same method as before and generate the point
$\bfy_{k+1} =  \bfy_k + \bfv + \bfw_k$.
In our implementation, we typically chose $\sigma$ comparable to $\texttt{tol}$.  This could occasionally produce steps that are larger than $\texttt{tol}$, so for a strict tolerance on the step size one might wish to truncate the larger steps.

\item \emph{Sampling}. We sample points from the density
\begin{equation}\label{eq:rho}
\rho(\bfx) = \frac{1}{Z} \exp(-\beta U(\bfx)), \qquad
\text{where }\;\;
U(\bfx) = \left\{ \begin{array}{cl}
|\bfx-\bfx_1|^2 & x \in M \\
\infty & x \notin M
\end{array} \right.
\end{equation}
where $\beta \in \R$ is a real parameter and $Z>0$ is a normalization constant. We sample on the manifold from the density $\rho$ using the Markov Chain Monte Carlo algorithm described in \cite{mcmc} and a step size parameter $\sigma$, again usually comparable to $\texttt{tol}$. The parameter $\beta$ is called the ``inverse temperature" in simulated annealing or other sampling techniques \cite{liu}, and it controls how peaked the density is near the minimum of $U(\bfx)$: large $\beta$ means the density is strongly peaked near the minimum so sampling pushes points toward $\bfx_1$, $\beta\approx 0$  means the density is relatively flat so the manifold is sampled nearly uniformly (similar to the first method but with a Metropolis step to ensure we sample the correct density), and $\beta <0$ means the density is lowest at the optimum, so sampling should push us away from $\bfx_1$ in general. We found that $\beta <0$  helped to overcome boundary obstacles, as we describe later in our numerical experiments. 

\end{enumerate}

We remark that since steepest descent is like sampling with $\beta\to\infty$, our method of switching between steepest descent and random sampling is basically a form of simulated annealing with temperature cycling. 

Note that another option would be to use an active set method \cite{nocedal}, and add in additional equations when one hits a boundary to navigate along the boundary while still performing steepest descent. We leave this option for future work. One may still have to resort to a random method for an arbitrary manifold with an arbitrary boundary.

Our algorithm terminates when one of two conditions is met. One, if ever $|\bfy_k-\bfx_1|<\texttt{tol}$, we stop, and declare that we have found a path. Two, if ever the total number of points exceeds some maximum $\texttt{Nmax}$, we stop, and declare there is no path. 

\medskip

To summarize, our path-finding algorithm consists of the following steps: 
\begin{enumerate}[(i),nosep]
\item Generate points according to steepest-descent algorithm, as in \eqref{eq:yk}. 
\item If one of conditions (a-c) is met, switch to random mode, and generate \texttt{Nr} random steps as in \eqref{eq:dens} or \eqref{eq:rho}. 
\item Stop if either (a) $|\bfy_k-\bfx_1|<\texttt{tol}$ (declare path found), or total number of points exceeds  $\texttt{Nmax}$ (declare no path.) 
\end{enumerate}
\bigskip

The algorithm could fail, either by failing to find a path that exists (false negative), or by finding a path that doesn't actually exist (false positive.) More commonly is to fail to find a path when one exists.  Similar to the problem of finding low-energy paths between points on an energy landscape \cite{Bolhuis:2002ew}, no known algorithm can guarantee to find paths between points on arbitrary nonlinear manifolds, with arbitrarily complicated boundaries to navigate; certainly the longer one looks, the more likely one is to find a path, but we don't expect to ever be able to provide any guarantees. 
It is helpful that we can check a posteriori if $\T_{\bfx}$ is a group, as then we can fill in missing elements if the number of such elements is small. 

The algorithm can also find a path that doesn't exist in the continuum problem, for example if the numerical path jumps between disconnected components of the manifold, a possibility if the components approach closer than the continuity parameter \texttt{tol}. We never found an example where this happened, but it is possible with inappropriate parameter values or a particularly complicated manifold geometry.

\begin{figure}[!t]
\centering
\includegraphics[scale = 0.14]{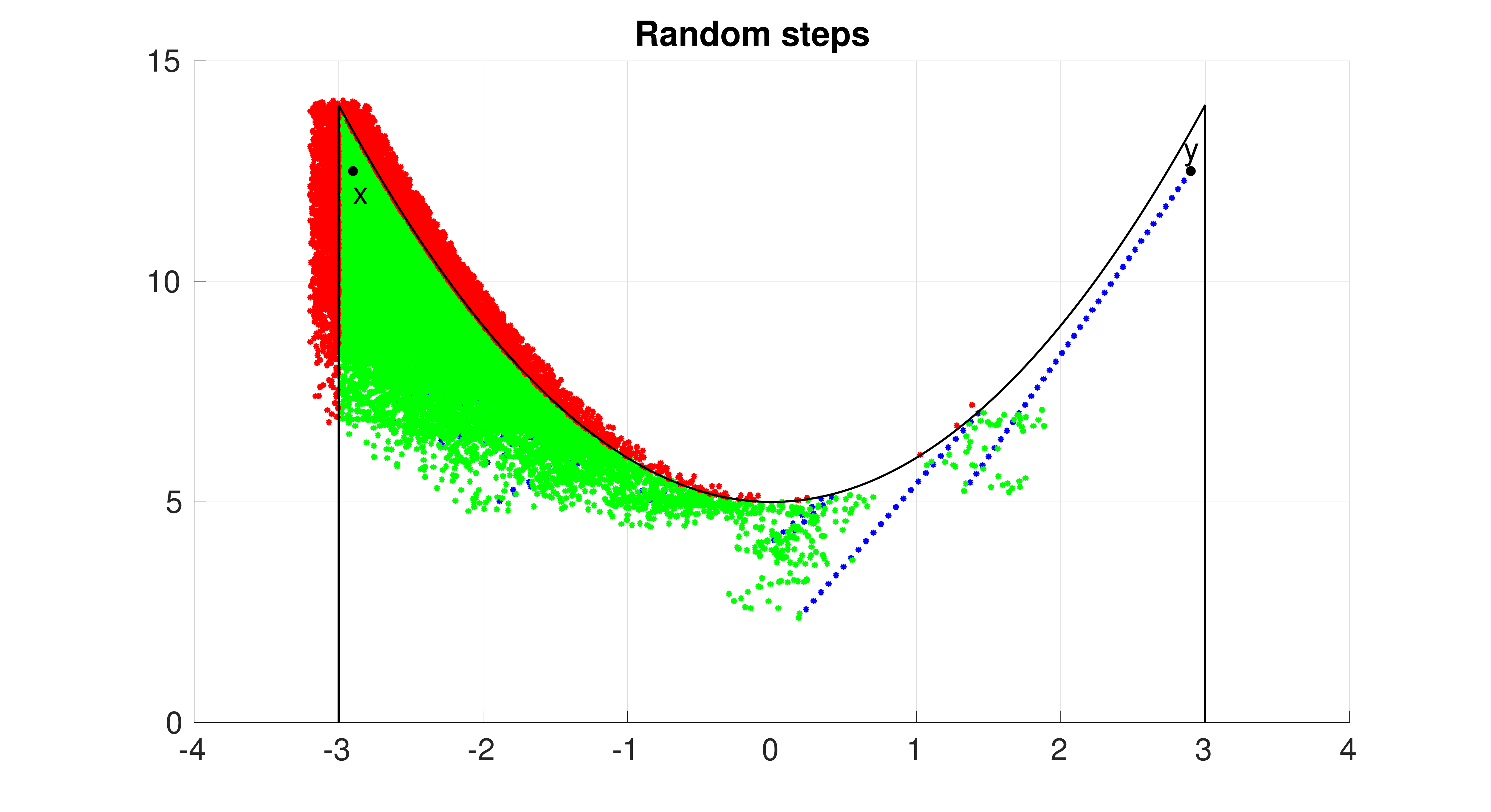}
\includegraphics[scale = 0.14]{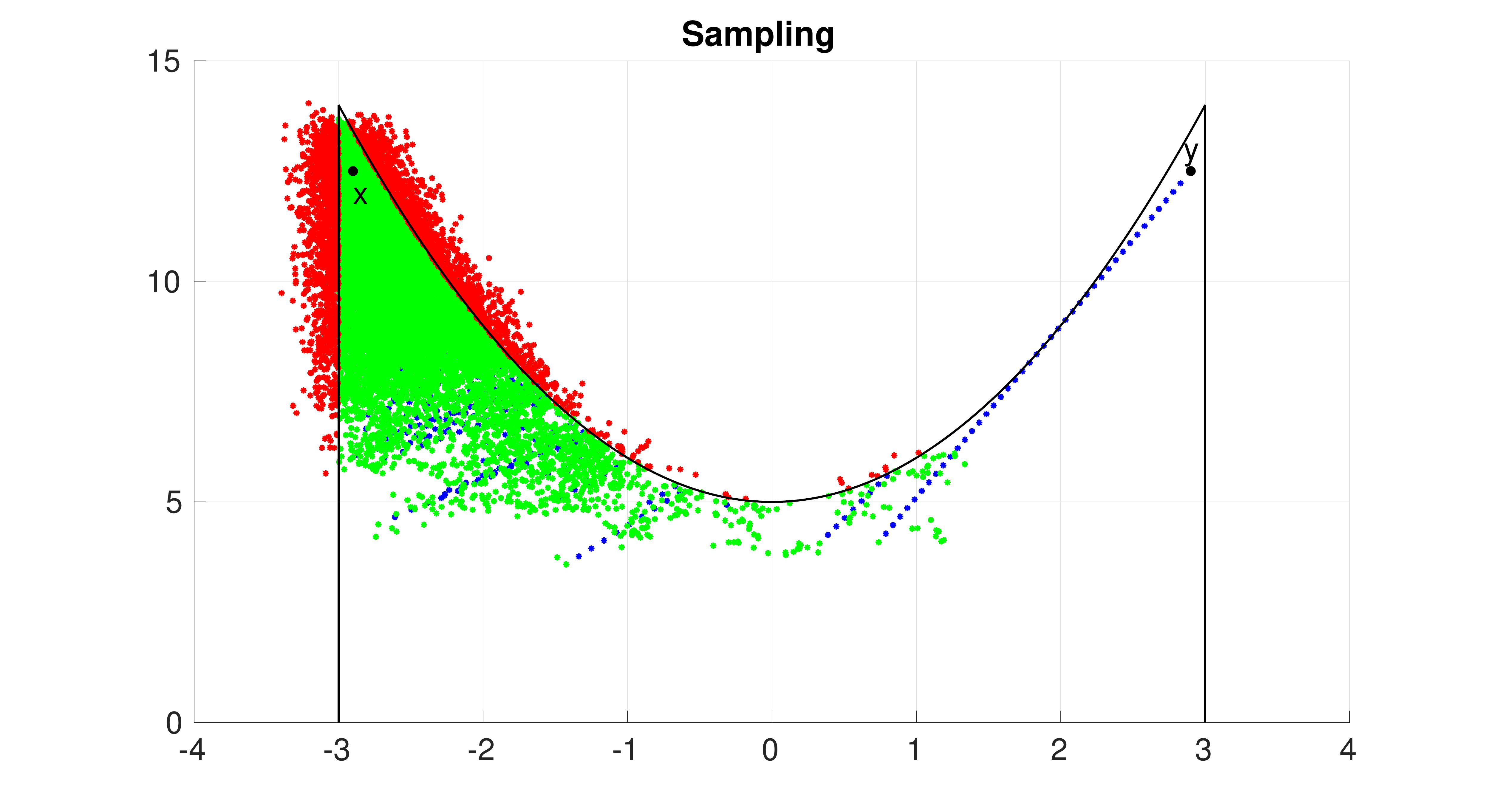}
\caption{Examples of a path connecting two points $\bfx,\bfy$ in the set $D \subseteq \R^2$ given in \eqref{eq:d}, constructed using the algorithm described in Section \ref{sec:num}. The blue points indicate the points obtained using the steepest descent method, while the green ones are generated with a random method: (left) using a Gaussian density and (right) using random sampling according to the distribution in \eqref{eq:rho}. Red points are rejected since they fall outside of the boundary. The total number of points generated is 21,530 (left) and 12,572 (right). The parameters used were $\Delta s = \sigma = \texttt{tol} = 0.2$, $\beta = -0.1$ and $\texttt{Nr} = 50$.   }
\label{fig:path}
\end{figure}

\paragraph{A toy example in $\R^2$.} 

We apply our algorithm to a toy example to illustrate how it works. We consider the set 
\begin{equation}\label{eq:d}
D = \{ (x,y) \in \R^2 : |x|<3, \; 0<y<x^2+5 \}.
\end{equation}
This set is visualized in Figure \ref{fig:path}. We fix the points $\bfx = (-2.9, 12.5)$ and $\bfy = (2.9, 12.5)$, and apply our algorithm to find a path in $D$ connecting $\bfx$ with $\bfy$. The path has to navigate around a boundary, moving in the opposite direction to the steepest descent direction $\bfy-\bfx$ for quite some time, to reach $\bfy$. 
We tested both methods for generating random points, and the resulting paths are shown in Figure \ref{fig:path}. Notice that the sampling with negative parameter $\beta$ results in a path with significantly fewer step points than the method with random steps. This is because the points are ``pushed away" from $\bfy$, and hence from the boundary $\partial M$, resulting in the sequence hitting the boundary fewer times.

\section{Examples}\label{sec:examples}

In this section we apply the theory developed so far to several examples, both to show that it works  and to analyze how well, as well as to point out pedagogical examples that illustrate properties of the sticky symmetry group. In the following, unless specified, we assume that the particles are spheres with diameter $d = 1$. 
In addition, we remove the three translational degrees of freedom from each cluster by fixing the center of mass to the origin, adding the additional three constraints $\sum_{i=1}^N \bfx_i = \mathbf{0}$ to the manifold \eqref{eq:m}. This doesn't change any of the theory regarding the symmetry number or how to compute it. 
In our implementation we usually set $\texttt{Nmax} = 10^4-10^5$.

\begin{figure}[!t]
\centering
\includegraphics[scale = 0.27]{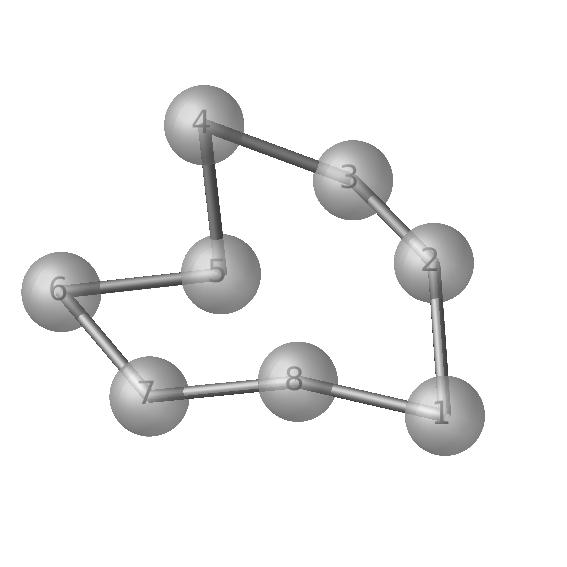}
\includegraphics[scale = 0.27]{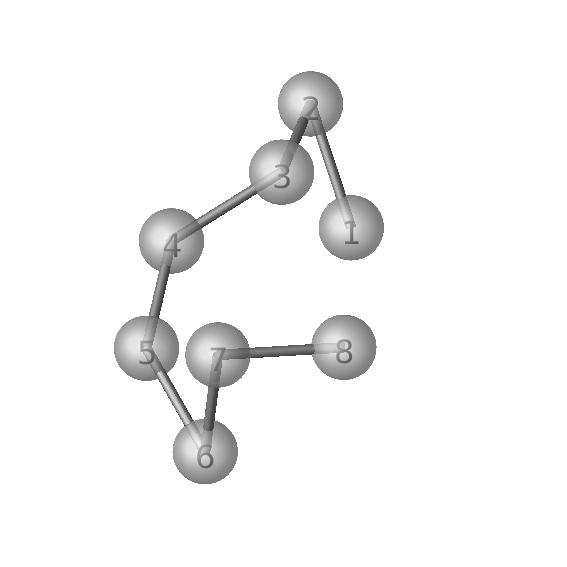}
\caption{A loop (left) and a chain (right) with $N = 8$ identical particles. The spheres are plotted with their radius half of the actual length and contacts are plotted as bars.}
\label{fig:loop_chain}
\end{figure}

\paragraph{Loops and chains.} We consider loops (L) and chains (C) of $N$ identical particles (see Figure \ref{fig:loop_chain}). A loop of $N$ particles has $N$ contacts, while a chain has $N-1$ contacts. These are flexible clusters. In particular, if we denote by $M_{L,N}$ and $M_{C,N}$ the manifold of configurations of a loop and a chain of $N$ particles, respectively, we have $\mbox{dim}M_{L,N} = 2N-3$ and $\mbox{dim}M_{C,N} = 2N-2$. 

We compute the automorphism group, the point group and the sticky symmetry group of loops and chains for  $N=4-10,15,20$. In order to test that our algorithm works, we randomly choose a point $\bfx$ in the manifold of configurations using the sampling algorithm in \cite{mcmc}, and compute $\Po_{\bfx}$ and $\T_{\bfx}$. In this way, $\bfx$ will have no ``a priori" symmetry axes.  For a loop of $N$ particles, the automorphism group $\G_{L,N}$ is isomorphic to the planar dihedral group $D_N$, the symmetry group of a regular $N$-gon in the plane. The point group of a random configuration $\Po_{L,N}$ is trivial (it consists of only the identity $(E, 1)$).
Using our algorithm, with $\sigma=\texttt{tol} = 0.1$, and $\texttt{Nr} = 20$ random steps, we find that $\T_{L,N} = \G_{L,N} \times C_2$, i.e. it is the entire automorphism-inversion group. One can visualize why this must be: starting from a configuration $\bfx \in M_{L,N}$, the loop can deform continuously until it reaches its most symmetrical configuration, a polygon that lies in a plane, and then it rotates according to the permutation $ P \in \G_{L,N}$ in the automorphism group, and finally it deforms again until it reaches the new configuration $\pm \widetilde{P}\bfx$. The symmetry number of a loop of $N$ particles is then $\sigma_{L,N} = |\T_{L,N}| = 4N$, and the counting number is  $n_{L,N} = 2N!/(4N) = (N-1)!/2$.

For a chain of $N$ particles, the automorphism group $\G_{C,N}$ consists of two elements: the identity $E$ and the permutation 
\begin{equation*}
\pi = \left\{ 
\begin{array}{cl}
(1 \; N)(2 \; N-1) (3\;N-2) \ldots \left( \frac{N}{2}  \; \frac{N}{2}+1 \right) & \text{if $N$ is even} \\
(1 \; N)(2 \; N-1) (3 \; N-2) \ldots \left( \frac{N-1}{2}+2\right) & \text{if $N$ is odd}
\end{array} 
\right. 
\end{equation*}
which corresponds to a reflection along the axis passing through the center of mass of the chain, when the chain assumes its most symmetrical configuration, with the spheres lying in a line. The point group of a random configuration is trivial, and again applying our algorithm we find that $\T_{C,N} = \G_{C,N} \times C_2$. Therefore, the symmetry number is $\sigma_{C,N} = 4$ and the counting number is $n_{C,N} = 2N!/4 = N!/2$.

\begin{figure}[!t]
\centering
\includegraphics[scale = 0.23]{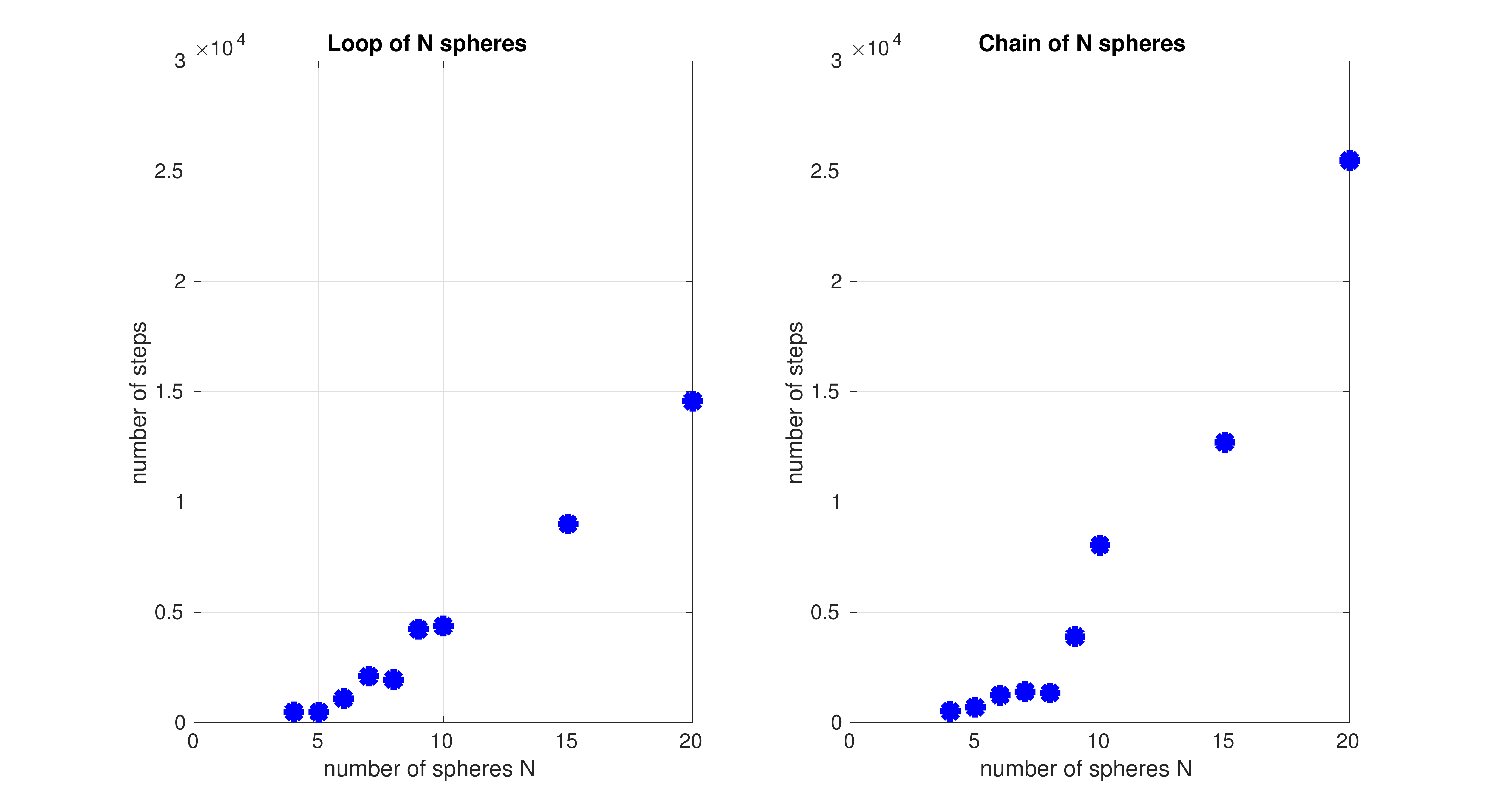}
\caption{Average number of steps for finding a path connecting $\bfx$ with $\pm \widetilde{P}\bfx$, where $P \in \T_{\bfx}$, for loops (L) and chains (C) of $N$ particles.}
\label{fig:steps}
\end{figure}

In Figure \ref{fig:steps} we plot the average number of steps to find a path for loops and chains of different sizes, where the average is over all paths that were successfully found, i.e. all elements in the corresponding sticky symmetry groups (we didn't consider variations with the initial condition or with different realizations of the noise.) The average number of steps increases roughly linearly with $N$, though possibly faster than linearly for chains. 

 We point out that, in the context of molecular symmetry, the loop of $N = 6$ particles corresponds to the benzene molecule.  Our algorithm finds the symmetry number $\sigma_{L,6} = 24$ for the $6$ loop, in agreement with previous results \cite{symmetry1}, where the symmetry number of benzene is computed using GAP\footnote{Actually, the symmetry number for the benzene in \cite{symmetry1} was 12, since the authors only consider permutations, and not permutation-inversion operations.}.

\begin{figure}[!t]
\centering
\includegraphics[scale = 0.23]{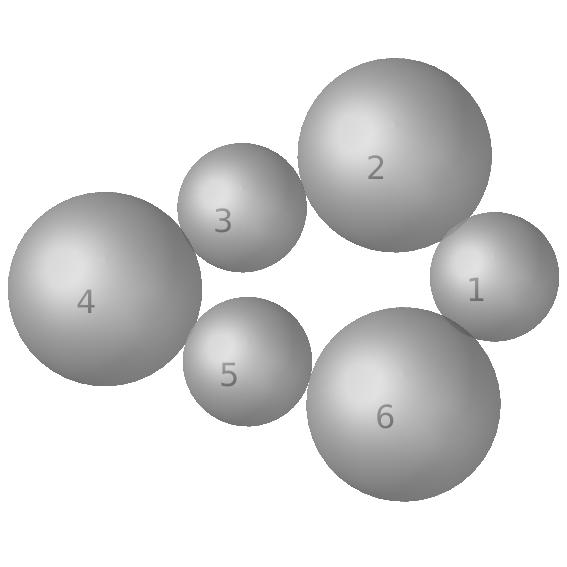}
\caption{A cluster with $N = 6$ spheres with different radii. The spheres $2,4$ and $6$ have radius $r = 0.6$, while the spheres $1,3$ and $5$ have radius $r' = 0.4$. The sticky symmetry group is generated by the rotation $(135)(246)$ and the reflection $(13)(46)$ and consists of six elements.  }
\label{fig:radii}
\end{figure}

\paragraph{Cluster of particles with different radii.} 
We consider again a loop of $N = 6$ particles, but this time we set the radius of the particles $2,4$ and $6$ to be $r = 0.6$, and the radius of particles $1,3$ and $5$ to be $r' = 0.4$ (see Figure \ref{fig:radii}). In this case, we first of all check which element of the automorphism group $\G_{L,6}$ of the loop (with identical particles) preserves the partition $R$ of the vertices according to the radii (compare with \eqref{eq:pnc}). We find that the automorphism group $\G^\mathcal R$ of the cluster is  
\begin{equation*}
\G^\mathcal R = \{ (135)(246), (153)(264), (26)(35), (13)(46), (24)(15), E \}. 
\end{equation*}
In particular, $\G^\mathcal R$ is generated by the rotation $(135)(246)$ and the reflection $(13)(46)$, and it is isomorphic to the dihedral group $D_3$, the symmetry group of an equilateral triangle.  
Using our numerical algorithm (applied to the manifold of configurations with different radii as in \eqref{eq:m}), we find that the sticky symmetry group $\T_{\bfx^\mathcal R}$ corresponds to $\G^\mathcal R \times C_2$.

\paragraph{Cluster of $N = 6$ spheres with two bonds broken.} Consider clusters with $N = 6$ identical spheres. There are two rigid clusters with $m = 12$ contacts, the octahedron and the polytetrahedron \cite{arkus}. We consider all the clusters with $m=10$ bonds, formed by deleting two bonds from a rigid cluster (and keeping only clusters with nonisomorphic adjacency matrices.)
The dimension of each manifold is $18-10-3 = 5$. 
For each cluster we compute its symmetry and counting number. The results are given in Table \ref{2dmodes}, where each cluster is referred to as a ``mode",  adopting terminology from \cite{miranda}. In the computations, we used random steps to explore the boundary, and parameters $\sigma = \texttt{tol} = 0.1$. These results agree with the computations done in \cite{miranda} for the same clusters, where the symmetry number of flexible clusters is computed using a combinatorial argument\footnote{The counting number in \cite{miranda} is actually half the counting number reported here, since the reflections were not taken into consideration.}. 

\begin{table}[!t]
\begin{center}
\begin{tabular}{c c c}
Mode &  Symmetry number $\sigma_{\bfx}$ & Counting number $n_{\bfx}$ \\
\hline
$8$ & 4  & 360 \\
$9$ & 2 & 720 \\
$10$ & 4 & 360 \\
$11$ & 2 & 720 \\
$12$ & 4 & 360 \\
$13$ & 2 & 720 \\
$14$ & 8 & 180 \\
$15$ & 10 & 144 \\
$16$ & 2 & 720 \\
$17$ & 2 & 720 \\
$18$ & 2 & 720 \\
$19$ & 6 & 240 \\
$20$ & 8 & 180 
\end{tabular}
\end{center}
\caption{Symmetry and counting numbers for clusters of $N=6$ 
identical spheres with $m=10$ contacts. The numbering is the same as in \cite{miranda}, where the manifolds are called ``modes'', to facilitate comparison.}
\label{2dmodes}
\end{table}

\paragraph{A cluster of $N = 6$ spheres with $\Po_{\bfx} \subsetneq\T_{\bfx}\subsetneq \G \times C_2$.} 

We consider in detail a cluster $\bfx$ with $N = 6$ spheres and $m=10$ bonds among the ones analyzed above (mode 14 in \cite{miranda}.) It is plotted in Figure \ref{fig:cluster1}, where we show our labeling convention of the spheres. 
This is an example of a cluster where $\Po_{\bfx} \subsetneq\T_{\bfx}\subsetneq \G\times C_2$, and the inclusions are strict. 

Indeed, 
the automorphism group $\G$ consists of 16 elements (we write them as permutations of the labeled particles):
\begin{align*}
\G  = & \{ E, (14)(23)(56), (14)(23), (1234)(56), (1234), (13)(24)(56), (13)(24), (13)(56), \\
& (13), (1432)(56), (1432), (12)(34)(56), (12)(34), (24)(56), (24), (56) \},
\end{align*}
where $E$ denotes the identity element. 

We compute the point group $\Po_{\bfx}$ of the cluster by checking which of the permutations in $\G$ preserve the distance matrix $D_{\bfx}$ (compare with \eqref{eq:PG} and \eqref{eq:dist_mat}), and additionally determining the associated inversion. This consists of four permutation-inversions: 
\begin{equation*}
\Po_{\bfx} = \{E, (12)(34)^*, (56)^*, (12)(34)(56)\},
\end{equation*}
where the $*$ indicates that the permutation is combined with a reflection $(\delta = -1)$. 
This group is generated by one reflection, and one rotation along the axis passing through the particles $5$ and $6$ and the plane which contains the particles $1,2,3,4$. 

We compute the sticky symmetry group by checking which elements in $\G\times C_2 \setminus \Po_{\bfx}$ belong to the sticky symmetry group. For any $P \in \G\times C_2\setminus \Po_{\bfx}$, we use our numerical procedure to check if there is a continuous path in the manifold of configurations $M_A$ connecting $\bfx$ with $\pm \widetilde{P}\bfx$, where $A$ is the adjacency matrix of $\bfx$. We set $ \texttt{tol} = \sigma = 0.1$, and we adopt random sampling with $\beta = -0.1$. We find the sticky symmetry group $\T_{\bfx}$ consists of eight elements
\begin{equation*}
\T_{\bfx} =\{E, (12)(34)^*, (56)^*, (12)(34)(56), (14)(23)(56)^*, (14)(23), (13)(24)(56), (13)(24)^*\}.
\end{equation*} 

Therefore $\Po_{\bfx} \subsetneq\T_{\bfx}\subsetneq \G\times C_2$. Because there exists an element $P \in \G$ such that $(P,\delta) \notin \T_{\bfx}$, for any $\delta \in C_2$, the manifold $M_A$ is not connected and the disconnected components are not related by a reflection. For example, the permutation $(24)$ which switches particles $2$ and $4$ preserves adjacency, however it is not realizable as a continuous deformation or rotation or reflection of the cluster. 

\paragraph{A colored cluster.}
Next, we consider the same example as above but distinguish the particles using three different colors (Figure \ref{fig:cluster1}, right). We form the partition $\C =\{C_1,C_2,C_3\}$ with  $C_1 = \{1,2\}$, $C_2 = \{ 3,4\}$ and $C_3 = \{5,6\}$. We compute the sticky symmetry group $\T_{\bfx}^\C$ of the colored cluster by checking which elements in $\T_{\bfx}$ preserve the partition $\C = \{C_i : i =1,2,3\}$ (compare with \eqref{eq:part_p} and \eqref{eq:top_col}.) We have that $\T_{\bfx}^\C$ coincides with the point group $\Po_{\bfx}$ of the original cluster $\bfx$, since all the elements of $\Po_{\bfx}$ preserve the partition into colors. 

We point out that, in the case of colored particles, we only need to compute the sticky symmetry group $\T_{\bfx}$ for identical particles once, then check which elements of it preserve the coloring. In other words, we can change the partition without having to check again the connectivity of the original manifold of configurations $M_A$. This is not true if we change the radii, because then the manifolds themselves will change, so the paths need to be recomputed. 

\begin{figure}[!t]
\centering
\includegraphics[scale = 0.25]{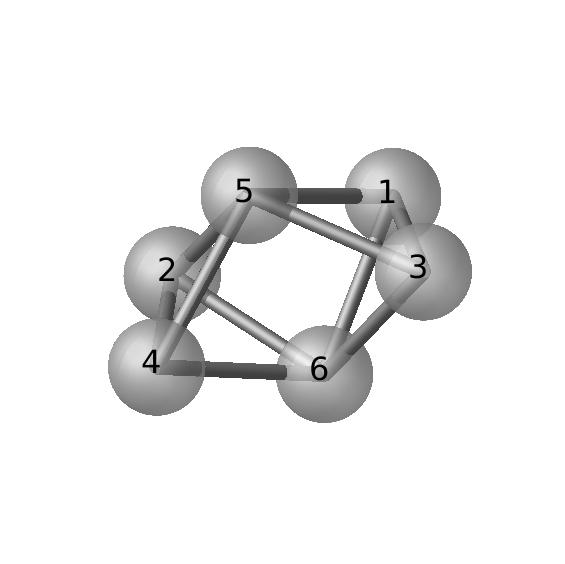}
\includegraphics[scale = 0.25]{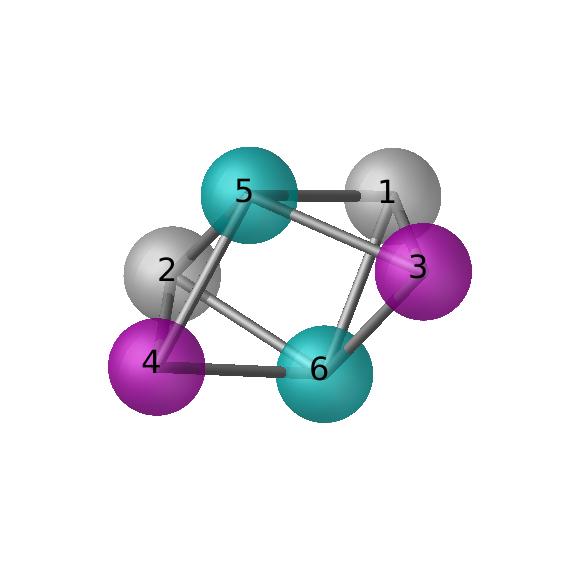}
\caption{Left: a symmetrical flexible cluster with $N = 6$ identical spheres. Right: the same cluster, with the particles distinguished using three different colors.  The spheres are plotted with their radius half of the actual length, to better visualize the bonds. }
\label{fig:cluster1}
\end{figure}

\begin{figure}[!t]
\centering
\includegraphics[scale = 0.29]{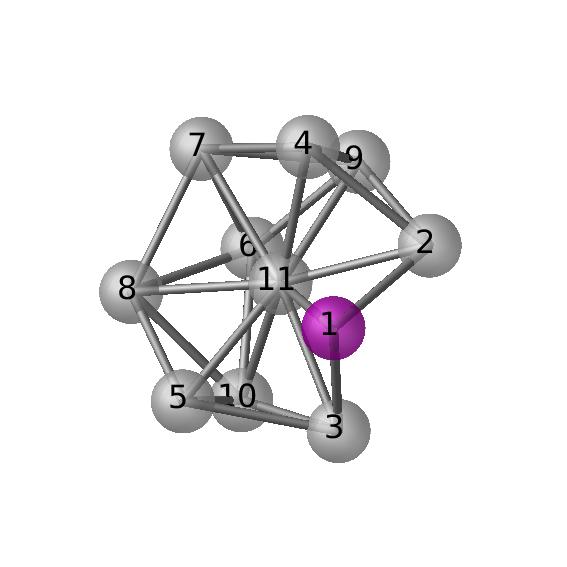}
\includegraphics[scale = 0.29]{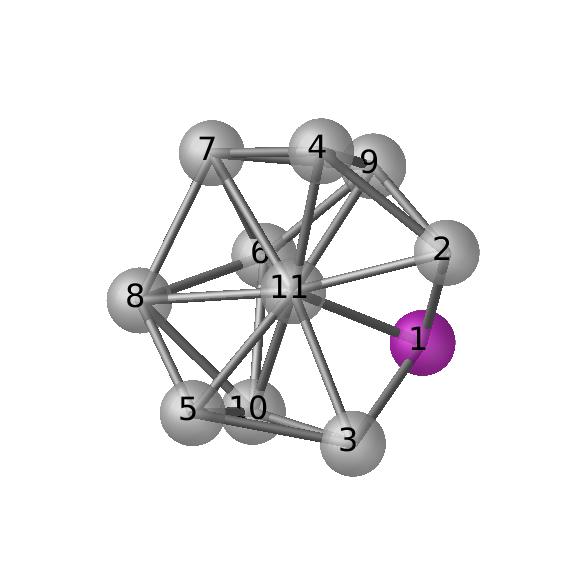}
\caption{Two rigid clusters with $N = 11$ identical particles which belong to different connected components of the same manifold of configurations. The gray spheres have identical coordinates, while the red particles (labeled 1) forms three contacts with spheres 2,3 and 11 in two different ways (see also \cite{siam}).}
\label{fig:11}
\end{figure}

\paragraph{Two clusters of $N = 11$ spheres where $|\T_{\bfx}|= |\G|$ but $M_{A}$ is disconnected.} 

We consider in detail two rigid clusters $\bfx_1$ and $\bfx_2$  (discovered in \cite{siam}) with $N = 11$ particles, which have the same adjacency matrix $A$, however there is no continuous transformation from one to the other or its reflection (see Figure \ref{fig:11}). This implies that they belong to different connected components of the manifold of configurations $M_A$, i.e. $M_{A,\bfx_1} \neq M_{A, \bfx_2}$.  

We use these clusters to build two flexible clusters with the same adjacency matrix, such that the sticky symmetry group has the same size as the automorphism group\footnote{
Note that we don't say they are the same because the sticky symmetry group includes inversion operations whereas the automorphism group does not; see \eqref{eq:incl}.
}, but that live on disconnected components of $M_A$. To do this we break one bond (bond $6-11$) in each cluster $\bfx_1$ and $\bfx_2$. The resulting two flexible clusters, which we denote by $\bfy_1$ and $\bfy_2$, have the same adjacency matrix $A'$ so they belong to the same manifold $M_{A'}$. However, we do not find a continuous path in $M_{A'}$ connecting $\bfy_1$ with either $\pm\bfy_2$, which implies that they belong to two different connected components of $M_{A'}$ that are not related by a reflection. The automorphism group $\G$ of $\bfy_1$ and $\bfy_2$ is given by 
\begin{equation*}
\G = \{E, (23)(45)(78)(9 \; 10)\}.
\end{equation*}
Using our algorithm, we find that $\T_{\bfy_1}=\T_{\bfy_2}$, and every element in $\G$ induces a single permutation-inversion operation in the sticky symmetry groups $\T_{\bfy_1},\T_{\bfy_2}$. 

Therefore, the manifold $M_{A'}$ is not connected, even if $|\T_{\bfy_1}|=|\T_{\bfy_2}|=\G$.   
Note that, in contrast to the previous example, we expect the two distinct components to be nonisomorphic, since clusters on each component are not related even by a permutation-inversion operation; we verified this heuristically by building physical models of the clusters with balls and magnetic sticks.

\paragraph{Partition functions of every connected cluster of $N=6$ identical spheres}

\begin{figure}
\begin{center}
{\footnotesize
\begin{tabular}{c c c}
$d$ & \# of clusters &  $Z_d$  \\\hline
0 & 2 & 3.9 \\
1 & 5 & 28\\
2 & 13 & 128\\
3 & 19 & 574\\
4 & 22 & $2.2\times 10^{3}$\\
5 & 19 & $4.7\times 10^{3}$\\
6 & 13 & $7.5\times 10^3$\\
7 & 6 & $7.2\times 10^3$\\\hline
\end{tabular}
}\hfill
\raisebox{-2.75cm}{\includegraphics[width=0.65\linewidth, trim=8cm 0cm 8cm 0cm, clip]{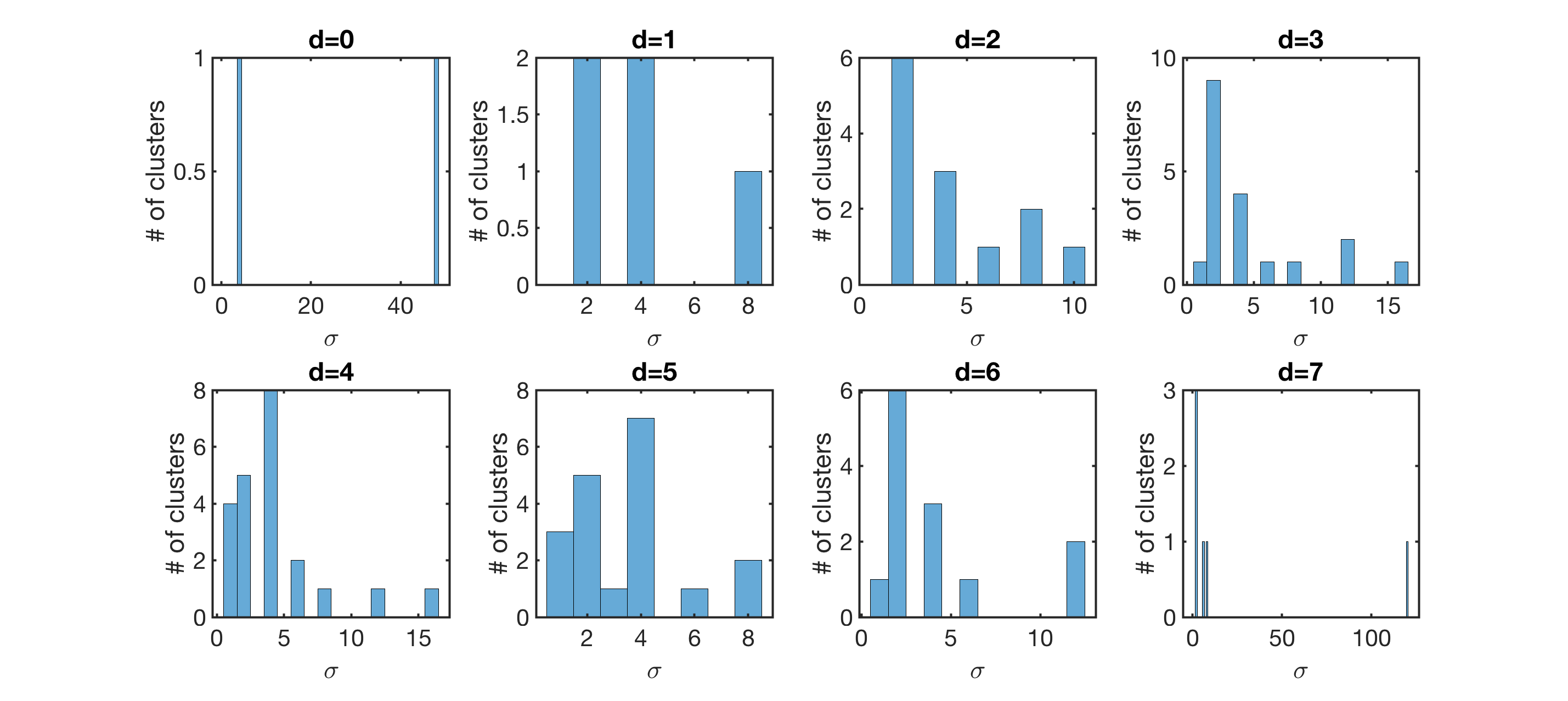}}
\end{center}
\caption{Left: Summary of the partition function calculations for $N=6$ identical spheres. For clusters with $d$ internal degrees of freedom, the table reports the total number of clusters, and the total geometrical partition function $Z_d$, equal to the sum of the geometrical partition functions of all clusters with dimension $d$.
Right: histograms of symmetry numbers for clusters of each dimension.}
\label{fig:n6}
\end{figure}

\begin{figure}
\begin{center}
\includegraphics[width=0.45\linewidth]{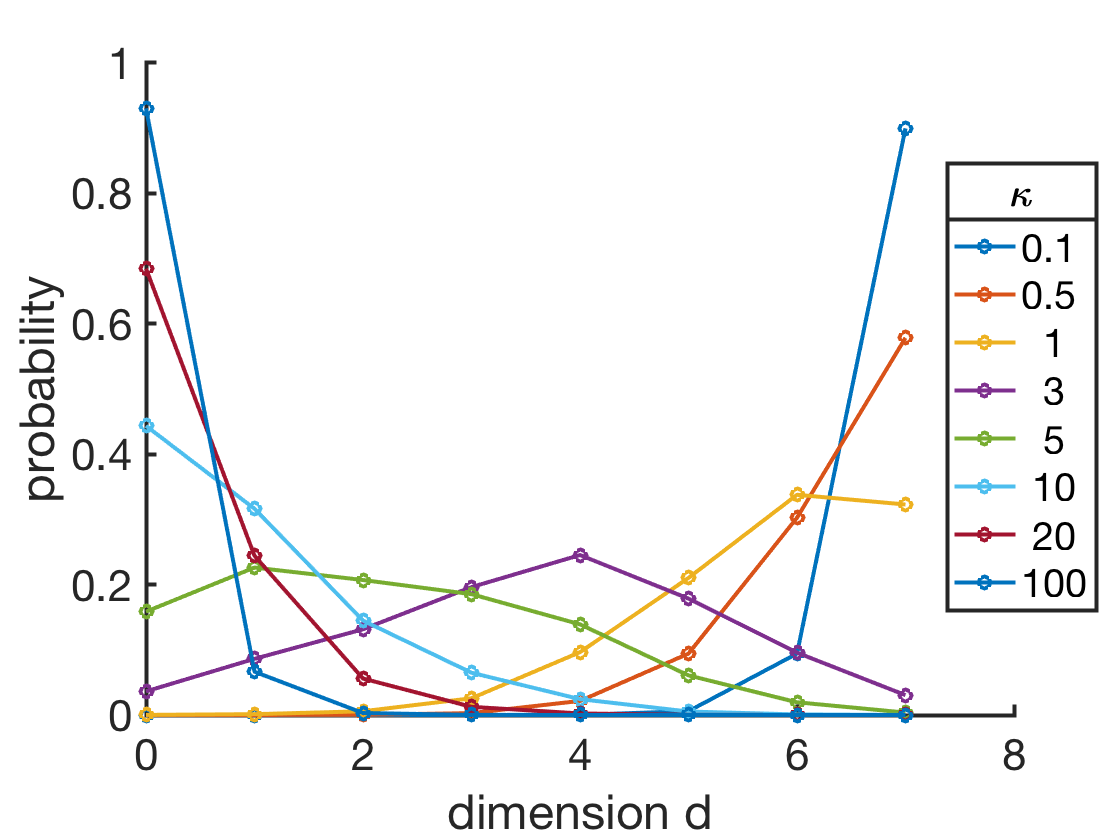}\hfill
\includegraphics[width=0.45\linewidth]{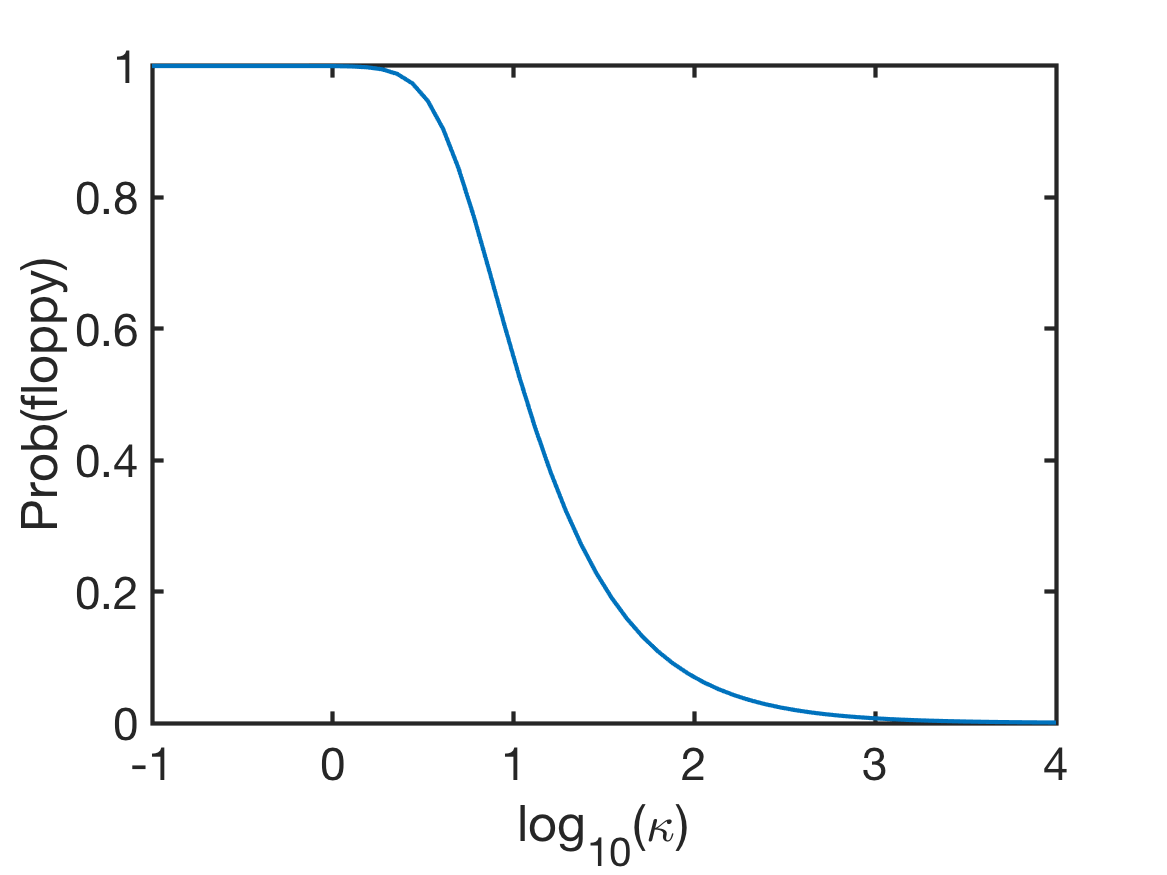}
\end{center}
\caption{Left: equilibrium probability of observing a system of $N=6$ identical spheres in a cluster with $d$ internal degrees of freedom, conditional on the cluster being connected, at different values of the sticky parameter $\kappa$. 
Right: equilibrium probability of finding the system in a floppy configuration ($d>0$) as a function of $\kappa$. 
}\label{fig:n6b}
\end{figure}

As a final example we compute the partition function \eqref{Zx} of every connected cluster of $N=6$ indistinguishable sticky spheres, at every temperature. With these partition functions in hand, one can determine essentially any equilibrium quantity one wants at any temperature for a system of 6 spheres (conditional on the spheres being connected), and, with further similar computations, one could begin to ask about non-connected systems, non-identical spheres, and spheres with non-identical interactions. Performing such an exhaustive calculation of the free energy landscape for sticky particles is an accomplishment in itself. Such an exhaustive calculation is not possible for non-sticky particles, i.e. those with smoother interaction potentials, where free energy landscapes are typically characterized by a collection of local minima and saddle points \cite{wales},  which vary with temperature and cannot be guaranteed to be found by existing numerical algorithms. 

We start with the two rigid clusters, which are known to be the only clusters with $m=12$ contacts \cite{Meng:2010gsa,siam}. On each cluster, we break, in turn, each bond, each pair of bonds, each triple of bonds, etc. For each graph $\alpha$ so obtained we check that it is connected, and that it is not isomorphic to a graph we have already seen.\footnote{For larger systems one would also need to check that each isomorphic graph is also on the same connected component of $M_{A,\bfx}$ (after applying the permutation), but we don't anticipate disconnected copies to be a problem in such a small system.} If it passes both tests, then we compute the integral in \eqref{Zx} using the method described in \cite{mcmc} (with center of mass of the cluster fixed), which is essentially a form of thermodynamic integration on a manifold, and call the resulting quantity $I_\alpha$. Then we calculate the symmetry number $\sigma_\alpha$ using the method described in this paper. The so-called ``geometrical'' partition function for graph $\alpha$ (i.e. the partition function without the factors of $\kappa$) is, up to a constant that is the same for all clusters under consideration, 
$z^g_\alpha = I_\alpha/\sigma_\alpha$. 

The table in Figure \ref{fig:n6} reports the total number of clusters found for each dimension $d$, where $d$ is, equivalently, the number of bonds broken from a rigid cluster, the number of internal degrees of freedom of the cluster, or the dimension of the corresponding manifold $M_{A,\bfx}$ in \eqref{Zx} (minus 3, to account for rotational degrees of freedom.) 
The number of clusters increases up to $d=4$, and then decreases to $d=7$, for a total of 99 clusters; for $d>7$ there are no  connected clusters. 
This table also reports the total geometrical partition function $Z_d$ for clusters with dimension $d$, formed by summing the geometrical partition functions $z^g_\alpha$ over all graphs $\alpha$ with dimension $d$. 
The ratios $Z_1/Z_0 = 7.1$, and $Z_2/Z_1 = 4.6$, agree within numerical error with those reported in \cite{miranda} (7.1, 4.5 respectively), partially verifying that our computations are correct; it is worthwhile to note that the methods used to compute $Z_0,Z_1,Z_2$ in \cite{miranda} cannot be extended to higher-dimensional clusters. 

Figure \ref{fig:n6} also shows the histograms of symmetry numbers for clusters of each dimension. The most symmetric cluster, at $d=7$, has a symmetry number $\sigma=120$. This is the cluster formed when 5 spheres touch a central sphere with a single contact each. The elements in the symmetry group are all of the $5!$ permutations of the outer spheres. 

The actual partition function for any cluster with graph $\alpha$ is obtained from the geometrical partition function as $z_\alpha = \kappa^m z^g_\alpha$, where $m=3N-d$ is the number of contacts in the cluster, and $\kappa$ is the sticky parameter, which measures how strong the bonds are between spheres: large $\kappa$ means spheres like to spend more time sticking together; $\kappa$ small means they come apart easily. For an experimental system where the pair interactions are not perfect delta functions, $\kappa$ would be a function of temperature, and the width and depth of the actual attractive interaction potential between the particles \cite{miranda}. It is large when temperature is small, and/or when the pair potential is deep and wide. Experimentally, clusters tend to form and rearrange when $\kappa \approx O(10)-O(100)$; any smaller, and the system evaporates, and any larger, and it takes a very long time for clusters to rearrange \cite{Perry:2015ku}. 

From the partition functions one can calculate the equilibrium probability to find the system in any given configuration. For example, 
Figure \ref{fig:n6b} (left) shows the equilibrium probabilities of finding the system in a cluster of a given dimension, at different values of $\kappa$ (recall this probability is conditional on the cluster being connected), equal to  $\kappa^{18-d}Z_d / Z$, where $Z=\sum_{d=0}^7 \kappa^{18-d}Z_d$ is the total partition function for the system. For large $\kappa$, the system is most likely to be in the lowest-dimensional configurations (those with the most contacts), while for small $\kappa$, it is likely to be in the highest-dimensional ones. For $\kappa\approx 1-7$, the probability has a maximum at an intermediate value of $d$; a range of $\kappa$ that coincides with the range of ratios $Z_{d+1}/Z_d$ (from smallest to largest $d$: 7.1, 4.6, 4.5, 3.8, 2.2, 1.6, 1.0.) 

Figure \ref{fig:n6b} (right) shows the probability of finding the system in a floppy configuration ($d>0$) as a function of $\kappa$. This probability is very close to 1 for small $\kappa$, and very close to $0$ for large $\kappa$, however it has a notably wide transition region, $\kappa\approx 4-70$, where the probability is not close to either endpoint (between 0.1-0.9); it crosses $0.5$ at $\kappa\approx 11.5$. In this transition region one would expect interesting dynamics, with the system forming clusters but rearranging substantially on observable timescales; hence, this is the range of $\kappa$ values one should aim for experimentally.

\section{Conclusion}\label{sec:conclusion}

We developed a theoretical and computational framework to compute the symmetry number of flexible sticky-sphere clusters, i.e. hard spheres interacting with a delta-function attractive interaction potential. 
We started from an equivalence relation which says that two clusters are the same if they are related by any combination of rotations, translations, or deformations that don't change the sphere-sphere contacts, 
 transformations which are all continuous and preserve the energy of a sticky-sphere cluster. 
 We showed how to count the number of distinct equivalence classes that one obtains by considering all permutations and reflections of a given cluster, a number we called the counting number. The counting number is related to a cluster's symmetry number, which is in turn obtained from its sticky symmetry group. 
We analyzed the sticky symmetry group and showed how it is related to two other groups commonly used to study molecular symmetries, the point group and the automorphism group of the graph describing the cluster's pairwise contacts. 

We introduced a numerical algorithm to compute the sticky symmetry group of a cluster. The key part of the algorithm is finding continuous paths in the manifold of configurations of the cluster. For this, we provided a numerical procedure alternating between steepest descent and random sampling, a form of simulated annealing. The algorithm comes with no guarantees, and could produce both false positives or false negatives, however we found with the right choice of parameters it worked extremely well for small clusters. In addition, we can use the fact that the collection of paths forms a group, to catch rare false negatives, i.e. paths that aren't found by the algorithm. An interesting question for future research, would be for what false negative rate is it possible to compute the entire group with high probability. 

We tested our algorithm on small clusters, some with up to $N=20$ spheres. 
Our algorithm is efficient when it is possible to compute the automorphism group $\G$, and gives the symmetry number many times faster than any hand computation would. 
For larger clusters or those with very high symmetry, several problems can occur. First, the automorphism group $\G$ can be prohibitively large to compute; for example \cite{symmetry3} gives an example of a cluster of 33 particles whose automorphism group has more than $10^{35}$ elements. Nevertheless, it may be possible to overcome this challenge by adapting algorithms for rigid molecules that cleverly avoid  computing the entire automorphism group \cite{symmetry2}. Secondly, as a cluster becomes larger, the geometry of its manifold can become more complicated, so it can require longer to find paths along it. 
Finally, for large clusters, there are a large number of constraints defining the manifold \eqref{eq:m}, so computing the tangent space, and projecting back to the manifold after taking a step in the tangent space, are more computationally expensive.

Our method can be used as an essential component of an algorithm that computes the entire free energy landscape of small clusters. As a step in this direction, we computed the partition functions for every possible connected cluster of $N=6$ identical spheres. One could perform a similar calculation for any $N\leq 8$, and for spheres with different radii, however for identical spheres larger $N$ would require dealing with contact constraints that are not linearly independent, leading to singularities in the manifolds of configurations that are challenging (though not impossible) to deal with numerically \cite{clusters,Kallus:2017hi}. Of course, for large enough $N$ one cannot exhaustively enumerate every configuration, but nevertheless symmetry or related topological considerations can be important in detecting different regimes of behaviour in systems of hard particles \cite{Carlsson:2012fo,Martiniani:2016bt}.


With such data, one can then ask how the free energy landscape varies as one varies the pairwise interactions between spheres. 
This question is important in materials science, where one might want to design particles, such as colloids, to assemble into a particular cluster, by varying the strengths and specificity of the interactions \cite{hormoz,Zeravcic:2014it}, control that is possible by coating colloids with strands of sticky DNA \cite{Wang:2015ep,Rogers:2016bd}.  In the sticky-sphere limit, such variation in interactions would correspond to changing the colorings of the particles and the strengths of the interactions between colors, represented by the coloring partition $\C$ and the sticky parameters $\kappa_{ij}$ where $i,j$ are colors (see \eqref{Zx}, and discussion thereafter.) 
Our method leads to an efficient way to solve this problem, since 
once we have calculated the integrals in \eqref{Zx} and the sticky symmetry groups $\T_{\bfx}$ for each connected component $M_{A,\bfx}$ of the landscape, we may 
obtain the partition functions for all clusters, for any coloring of the particles and any interaction strengths between colors, without repeating these arduous computations. 
Therefore, we expect the framework outlined in this paper to be an important component of an algorithm which aims to efficiently compute cluster probabilities over a wide range of interactions strengths and specificities, and to solve the inverse problem of asking how to make a particular cluster with high probability, given experimental constraints.

\paragraph{Acknowledgments.} 
We thank Louis Theran for useful discussions.  
This material is based upon work supported by the U.S. Department of Energy, Office of Science, Office of Advanced Scientific Computing Research under award DE-SC0012296.
M.H.-C. acknowledges support from the Alfred P. Sloan Foundation.

\section{Appendix: proofs}

In this Appendix we provide the proofs of some statements of Section \ref{sec:count}. 

First of all, while it seems like an obvious fact, we will need to know that if we apply a permutation-inversion $(P,\delta)$ to a cluster $\bfx$, we obtain a manifold $M_{PAP^T,\delta \widetilde{P}\bfx}$ which is an isometry of $M_{A,\bfx}$, i.e. it is geometrically indistinguishable from it. We prove this statement in the following lemma, which in addition will imply that continuous paths in $M_{A,\bfx}$ get mapped to continuous paths in $M_{PAP^T,\delta \widetilde{P}\bfx}$ after applying the permutation-inversion operation $(P,\delta)$ to each element along the path. 

\begin{lemma}\label{lem:isometry}
Let $\bfx \in \R^{3N}$ be a cluster of $N$ identical spheres. Let $\bfx' = \delta \widetilde{P}\bfx$, with $P \in P(N)$ and $\delta \in C_2$. Let $A'$ denote the adjacency matrix of $\bfx'$. The map 
\begin{equation}\label{eq:phi}
\begin{aligned}
\phi_{P, \delta} : M_{A, \bfx} & \longrightarrow M_{A',\bfx'} \\
\bfy & \longmapsto \delta \widetilde{P} \bfy
\end{aligned}
\end{equation}
is an isometry.
\end{lemma} 

\begin{proof}
The function $\phi_{P,\delta}$ is clearly smooth and bijective, with inverse $\phi_{P, \delta}^{-1} = \phi_{P^{-1}, \delta^{-1}}$. Let $T_{\bfx}M_A$ and $T_{\bfx'}M_{A'}$ denote the tangent spaces at $\bfx$ and $\bfx'$ of the manifolds $M_A$ and $M_{A'}$, respectively. 
Because $\phi_{P,\delta}$ is a linear function, 
the tangent map $T\phi_{P,\delta} : T_{\bfx}M_A \longrightarrow T_{\bfx'}M_{A'}$ may be easily calculated to be $T\phi_{P,\delta}(\bfv) = \delta \widetilde{P}\bfv$. 
We have, for every $\bfv, \bfu \in  T_{\bfx}M_A$, 
\begin{equation*}
\langle T\phi_{P,\delta}(\bfv), T\phi_{P,\delta}(\bfu)\rangle_{\bfx'} = \langle \delta \widetilde{P}\bfv, \delta\widetilde{P}\bfu \rangle = \delta^2 \bfv^T \widetilde{P}^T \widetilde{P}\bfu = \bfv^T \widetilde{P^TP} \bfu = \langle \bfv, \bfu \rangle_{\bfx},
\end{equation*}
where $\langle, \rangle_{\bfx}$ (respectively $\bfx'$) is the standard Euclidean metric in $\R^{3N}$ restricted to in $M_A$ (respectively $M_{A'}$.) This proves our claim. 
\end{proof}


This lemma allows us to prove that the function given in \eqref{eq:action} is in fact a well-defined action of $P(N) \times C_2$ on the quotient set $X$ given in \eqref{eq:Xr}. 

\begin{prop}
The function $\cdot : (P(N) \times C_2) \times X \longrightarrow X$ given by 
\begin{equation*}
(P,\delta) \cdot [\bfx] = [\delta \widetilde{P}\bfx] = M_{PA(\bfx)P^T, \delta \widetilde{P}\bfx}
\end{equation*}
is a well-defined action of $P(N) \times C_2$ on the quotient set $X$.
 \end{prop}
\begin{proof}
We need to show two things: one, that the action is well-defined on equivalence classes, and two, that it respects the property of multiplication. 

For the first, 
suppose $\bfx \sim \bfy$ are two equivalent clusters. We want to show that $\delta \widetilde{P}\bfx \sim \delta \widetilde{P}\bfy$ or, equivalently, that $[\delta \widetilde{P}\bfx] = [\delta \widetilde{P}\bfy]$, for each $(P,\delta) \in P(N)\times C_2$. We have 
\begin{equation*}
[\delta \widetilde{P}\bfx] = M_{PA(\bfx)P^T, \delta\widetilde{P}\bfx}, \qquad [\delta \widetilde{P}\bfy] = M_{PA(\bfy)P^T, \delta \widetilde{P}\bfy}.
\end{equation*}
Since $\bfx \sim \bfy$, we have that $M_{A(\bfx), \bfx} = M_{A(\bfy), \bfy}$. In particular, this implies $A(\bfx) = A(\bfy)$, and the existence of a continuous path $\varphi : [0,1] \longrightarrow M_{A(\bfx)}$ such that $\varphi(0) = \bfx$ and $\varphi(1) = \bfy$. We now consider the isometry $\phi_{P,\delta}$ as in \eqref{eq:phi} and construct the path 
\begin{equation*}
\psi(t) = (\phi_{P,\delta} \circ \varphi)(t) : [0,1] \longrightarrow M_{PA(\bfx)P^T}\,,
\end{equation*}
which is a continuous path in $M_{PA(\bfx)P^T}$ connecting $\psi(0) = \delta\widetilde{P}\bfx$ with $\psi(1) = \delta \widetilde{P}\bfy$. Therefore, $\delta \widetilde{P}\bfx$ and $\delta \widetilde{P}\bfy$ belong to the same connected component of $M_{PA(\bfx)P^T}$, implying $M_{PA(\bfx)P^T, \delta \widetilde{P}\bfx} = M_{PA(\bfy)P^T, \delta \widetilde{P}\bfy}$, which proves our claim. 

Finally, $\cdot$ is an action since, for every $(P,\delta)$, $(Q,\mu) \in P(N)\times C_2$,
\begin{equation*}
(P,\delta) \cdot \left((Q,\mu) \cdot [\bfx]\right) = (P,\delta) \cdot [\mu\widetilde{Q}\bfx] = [\delta \mu \widetilde{PQ} \bfx] = (PQ, \delta \mu) \cdot [\bfx] = ((P,\delta) (Q, \mu)) \cdot [\bfx],
\end{equation*}
and clearly $(I_N,1) \cdot [\bfx] = [\bfx]$, where $I_N$ is the $N \times N$ identity matrix. 
\end{proof}

Finally, we prove the connection between point group and distance matrix of a cluster, a result that is widely used but for which we have found no accessible proof in the literature, so we prove it here for completeness.  
A canonical (though incomplete) reference for this result is  \cite{YH38}, and \cite{edm} provides a clear explanation of distance matrices and some common manipulations with them. 

\begin{prop}
Let $\bfx \in \R^{3N}$ be a cluster of $N$ spheres, and let $\Po_{\bfx}$ be the point group of $\bfx$ as in \eqref{eq:PG}. Then 
\begin{equation}\label{eq:point_2}
(P,\delta) \in \Po_{\bfx} \text{ for some } \delta \in C_2 \;\;\Longleftrightarrow\;\; PD_{\bfx} = D_{\bfx}P. 
\end{equation}

\end{prop}

\begin{proof}
Let 
$B = (\bfx_1| \ldots |\bfx_N)^T$
%
 be the $N \times 3$ matrix whose rows are $\bfx_i$, for $i = 1, \ldots, N$. Let $G_{\bfx}$ be the $N \times N$ Gram matrix of $\bfx$, computed from $B$ as
 \[
  G_{\bfx} = BB^T\,.
 \] 
 We first prove that preserving the distance matrix is equivalent to preserving the Gram matrix, i.e. 
\begin{equation}\label{DG}
PD_{\bfx} P^T = D_{\bfx} \;\;\Longleftrightarrow \;\; PG_{\bfx}P^T = G_{\bfx},
\end{equation}
for every $P \in \G$. 
To see this, first notice that the matrices $G_{\bfx}$ and $D_{\bfx}$ are related by the formula 
\begin{equation}\label{eq:dx2}
D_{\bfx} = \text{diag}(G_{\bfx})\mathbf{1}^T -2G_{\bfx}+\mathbf{1}\text{diag}(G_{\bfx})^T,
\end{equation}
where $\mathbf{1}$ is a $N \times 1$ column vector of all ones \cite{edm}. 
Next, notice that $P(\text{diag}(G_{\bfx})\mathbf{1}^T)P^T = \text{diag}(G_{\bfx})\mathbf{1}^T$. In fact, the $i$-th row of the matrix $K = \text{diag}(G_{\bfx})\mathbf{1}^T$ is a $1 \times N$ vector with the same entries, i.e. of the form $(c_i, \ldots, c_i)$, for some $c_i \in \R$. Therefore, the permutation matrix $P$ permutes the rows of $K$, which are then permuted again using the inverse permutation $P^T$, resulting in the original matrix $K$.  
We use this and \eqref{eq:dx2} to compute
\begin{align}
PD_{\bfx}P^T 
&= \text{diag}(G_{\bfx})\mathbf{1}^T-2(PG_{\bfx}P^T) + \mathbf{1}\text{diag}(G_{\bfx})^T\,.\label{eq:dx3}
\end{align}
Subtracting \eqref{eq:dx2} from \eqref{eq:dx3} gives
\[
PD_{\bfx}P^T - D_{\bfx} = -2(PG_{\bfx}P^T -G_{\bfx} )
\]
from which implication \eqref{DG} is clear. 


We now shift the attention to the Gram matrix $G_{\bfx}$. Specifically, we want to prove that, given $P \in \G$, then $(P,\delta) \in P_{\bfx}$ for some $\delta\in C_2$, if and only if $PG_{\bfx}P^T = G_{\bfx}$.

Showing that a permutation in the point group preserves the Gram matrix follows by direct calculation. 
Suppose $(P,\delta) \in P_{\bfx}$. Then, by definition of $\Po_{\bfx}$ (see \eqref{eq:PG}), there exists $R \in SO(3)$ such that $\widetilde{P}\bfx = \delta (R \otimes I_N)\bfx$. Let $\bfy = \widetilde{P}\bfx$, and $C = (\bfy_1| \ldots |\bfy_N)^T$. Then $C = PB$. The Gram matrix $G_{\bfy}$ of $\bfy$ is 
$
G_{\bfy} = CC^T = (PB)(PB)^T = PG_{\bfx}P^T. 
$ 
On the other hand, since $\bfy = \delta(R \otimes I_N)\bfx$, then $C = \delta BR$. This implies 
$
G_{\bfy} = (\delta BR)(\delta BR)^T = BB^T = G_{\bfx}.
$ 
Therefore $PG_{\bfx}P^T = G_{\bfx}$.

For the other direction, suppose $P$ is such that $PG_{\bfx}P^T = G_{\bfx}$. Then  $(PB)(PB)^T = BB^T$,  so the clusters $\bfy$ (formed from the rows of $C=PB$) and $\bfx$ (formed from the rows of $B$) have the same Gram matrix. 

It remains to show that if two clusters have the same Gram matrix, then they are related by an orthonormal transformation. 
This is a result in linear algebra that we reproduce here. 
Let 
\[
B = U_1\Sigma_1 V_1^T, \qquad C = U_2\Sigma_2 V_2^T
\]
be the singular value decompositions of $B,C\in \R^{N\times 3}$, 
where $U_i\in \R^{N\times N}$, $V_i\in \R^{3\times 3}$, $\Sigma_i\in \R^{N\times 3}$, for $i=1,2$. 
We are given that $G_{\bfx}=BB^T=CC^T$ and so $U_1\Sigma_1\Sigma_1^TU_1^T = U_2\Sigma_2\Sigma_2^TU_2^T$. But this is an eigenvalue decomposition of the symmetric matrix $G_{\bfx}$, which is unique up to reordering of eigenvalues and up to eigenvalues that are the same. We may order the eigenvalues in order of decreasing absolute value, and therefore the diagonal elements of the diagonal matrices $\Sigma_1\Sigma_1^T$, $\Sigma_2\Sigma_2^T$ may be chosen to be the same, which implies, since the diagonal elements of $\Sigma_1,\Sigma_2$ are nonnegative, that $\Sigma_1=\Sigma_2$. 
For eigenvalues that are the same, we may choose any orthogonal basis for the corresponding columns of $U_i$ among the available eigenvectors, and therefore we may choose the bases such that $U_1=U_2$. Therefore, we may write $C=U_2\Sigma_2V_2^T=U_1\Sigma_1 V_2^T = U_1\Sigma_1V_1^TV_1 V_2^T =BV_1V_2^T = BQ$ where $Q=V_1V_2^T\in \R^{3\times 3}\in O(3)$, and the result is proven. \qedhere


\end{proof}

\bibliographystyle{unsrt}
\bibliography{symmetry}

\end{document}